\documentclass[3p]{elsarticle}

\usepackage{amsmath,amssymb,amsfonts}
\usepackage{graphicx}
\usepackage{booktabs}
\usepackage{multirow}
\usepackage{siunitx}
\usepackage{hyperref}
\usepackage{xcolor}
\usepackage{float}
\usepackage{amsmath}
\usepackage{amsthm}
\usepackage{float}
\usepackage{multirow}
\usepackage{array}
\usepackage{booktabs}
\usepackage{amssymb}
\usepackage{xcolor}
\usepackage{tikz}
\usepackage{mathrsfs}
\usepackage{subcaption}
\usepackage{caption}
\usepackage{algorithm}
\usepackage{algpseudocode}
\usepackage{multirow}
\usepackage{soul} 
\usepackage{color}

\usepackage{graphicx}
\usepackage{graphicx}
\newtheorem{remark}{Remark}
\newtheorem{theorem}{Theorem}

\usepackage[utf8]{inputenc}
\usepackage[T1]{fontenc}
\usepackage{amsmath,amssymb,amsthm}
\usepackage{mathtools}
\usepackage{booktabs}
\usepackage{array}
\usepackage{float}
\usepackage{enumitem}
\usepackage{booktabs}
\usepackage{hyperref}
\usepackage{cleveref}

\newcommand{\yd}{y_d}

\algnewcommand\algorithmicgiven{\textbf{Given:}}
\algnewcommand\Given{\item[\algorithmicgiven]}
\algnewcommand\algorithmicdesign{\textbf{Design requirement:}}
\algnewcommand\Design{\item[\algorithmicdesign]}
\algnewcommand\algorithmicstates{\textbf{Persistent states:}}
\algnewcommand\States{\item[\algorithmicstates]}
\algnewcommand\algorithmicinit{\textbf{Initialization} ($k=0$)\textbf{:}}
\algnewcommand\Init{\item[\algorithmicinit]}

\begin{document}

\begin{frontmatter}

\title{Data-driven Control with Real-time Uncertainty Compensation \\for Multi-Fuel Engines}
\author[umn]{Rajasree Sarkar}
\author[umn]{Arunava Banerjee\fnref{CA}}
\author[umn]{Sathya Aswath Govind Raju}
\author[umn]{Ishan Berk Altiner}
\author[umn]{Zongxuan Sun}
\author[arl]{Kenneth Kim}
\author[arl]{Chol-Bum Mike Keown}
\address[umn]{Department of Mechanical Engineering, University of Minnesota Twin Cities, Minneapolis, MN 55455, USA.}
\address[arl]{DEVCOM Army Research Laboratory, Aberdeen Proving Ground, MD, USA, 21005.}
\fntext[CA]{Corresponding author: Arunava Banerjee, Department of Mechanical Engineering, University of Minnesota Twin Cities, Minneapolis, MN 55455, USA. email: abanerje@umn.edu}

\begin{abstract}
Multi-fuel compression ignition (CI) engines offer superior power density and fuel flexibility. However, achieving consistent and optimal combustion phasing across a wide range of operating conditions remains a major challenge, particularly in the presence of modeling uncertainties. This paper presents a novel, data-driven real-time uncertainty compensation framework for combustion control in multi-fuel CI engines. The proposed approach introduces a pseudo-engine speed that enables dynamic adaptation of control inputs in response to uncertainty affecting the engine. To model the underlying combustion process, a Gaussian Process Regression (GPR) model is first trained on available input-output data, capturing the nonlinear and fuel-dependent behavior across varying operating conditions. Control inputs are then synthesized through model inversion of the learned GPR surrogate and augmented with a uncertainty compensator designed to mitigate deviations caused by dynamic variations in operating conditions and model inaccuracies. This integrated control strategy allows for real-time input corrections within a finite number of combustion cycles. Theoretical analysis establishes finite-time convergence guarantees for the proposed controller. Simulation results demonstrate that the proposed method steers the combustion phasing to the desired value in real-time, providing a scalable and adaptive control solution for multi-fuel CI engine operation.
\end{abstract}


\end{frontmatter}

\section{Introduction}
Hybrid electric propulsion systems incorporating compression ignition (CI) engines are increasingly recognized as a promising approach to improve performance across a range of applications. \cite{rendon2021aircraft}. However, developing robust and resilient control strategies for such systems remains a considerable challenge. This difficulty stems from the inherently complex nature of combustion in CI engines, involving the interaction of multi-fluid physics and chemical reaction in a complex, turbulent flow field. Accurate modeling of the chemical reactions is crucial for designing effective combustion control strategies that ensure stable engine operation. Key phenomena influencing combustion, such as droplet dynamics, fuel-air mixing, and evaporation, are highly nonlinear and sensitive to variations in fuel properties and operating conditions \cite{mcgann2020effect}. The complexity is further exacerbated in multi-fuel scenarios, where real-time changes in fuel characteristics introduce additional uncertainty. Moreover, combustion behavior is strongly affected by dynamic operating conditions, including engine speed, ambient pressure, and temperature. Thus, due to the inherent nonlinearity and complexity of the input–output relationships in combustion processes, physics-based modeling approaches become increasingly difficult to implement with high accuracy. As a result, conventional model-based control techniques fall short when applied to the problem of combustion phasing control in CI engines, particularly under uncertain and time-varying operating conditions.

To address the challenges associated with complex system dynamics, data-driven modeling have emerged as a compelling alternative to traditional physics-based methods \cite{kalogirou2003artificial, aliramezani2022modeling}. These methods alleviate the need for explicit physical knowledge which often limit the scalability and adaptability of conventional control strategies under varying operating conditions. A variety of machine learning-based modeling approaches have been explored in the literature, including neural networks, Gaussian Process Regression (GPR), and Support Vector Machines (SVM), among others. A comprehensive comparison of these methods is presented in \cite{ostergaard2018comparison}. 
Among these, GPR has garnered significant attention due to its ability to deliver high modeling accuracy with relatively limited training data, while also providing probabilistic uncertainty quantification across the input space \cite{seeger2004gaussian}. These properties make GPR particularly attractive for control applications involving uncertain nonlinear systems \cite{zhao2025probabilistic}. Notable applications include tracking control \cite{dai2024decentralized}, model predictive control \cite{zhu2023gaussian}, and reinforcement learning \cite{zhao2023probabilistic}. Within the domain of combustion phasing control, early works of GPR used in engine modeling appears in \cite{berger2011analysing} which showed promising utilization of data-driven approaches for modeling complex engine combustion process with a relatively simple model for computation. Recently, GPR-based modeling approach is presented in \cite{dong2022data} wherein a feedforward controller is designed by inverting the learned model to generate lookup tables (LUTs) across the engine operating space. This approach maintains computational efficiency, which is critical for practical implementation, while leveraging the modeling strengths of GPR. Similarly, recent work in \cite{banerjee2026data} demonstrates the utility of GPR in real-time fuel property estimation by inferring unknown cetane numbers during dynamic fuel switching. Despite the strengths of GPR, its performance heavily depends on the availability of rich and representative datasets that is especially a critical requirement for engine combustion modeling. Unlike simpler systems where moderate datasets may suffice, CI engines operate under a wide range of conditions involving multiple interdependent variables, such as load, speed, ambient conditions, and fuel properties. This results in a high-dimensional input space, where capturing the full dynamics demands carefully curated data that reflects the complex interactions among these factors. Inadequate or sparse data can lead to significant degradation in model fidelity and control performance, making data collection a major bottleneck in deploying data-driven control for such systems. For CI engines, acquiring such datasets experimentally is both costly and time-consuming. However, as an alternative, high-fidelity simulators such as those based on computational fluid dynamics (CFD) can be used to generate relevant training data \cite{huang2016diesel}. Now, in order to enhance robustness against model uncertainties that arise from such modeling inaccuracies, recent studies have explored online model update mechanisms, including both synchronous and asynchronous strategies, to adapt GPR models in real time \cite{umlauft2019feedback,wu2022safe,zheng2024safety}. While these methods improve adaptability to changing environments, their computational overhead makes them less practical for fast, high-dimensional systems like engine combustion. Moreover, many of these approaches assume a control-affine system structure and require full-state measurement, assumptions that rarely hold in practical engine control applications \cite{wang2021neural,wang2023model}. Some methods attempt to circumvent this by designing state observers for unmeasured states \cite{zhang2021adaptive}. However, identifying relevant states or even the appropriate system order in combustion processes is nontrivial. As a result, achieving reliable output regulation under partial observability remains a significant open challenge.

To address the limitations associated with developing GP-based reliable and robust control designs for practical applications, such as CI engines, this paper proposes a novel compensating mechanism termed the \textit{uncertainty compensator}. The proposed approach ensures that the real engine output tracks the desired value even when the engine is subjected to different fuel blends, varying engine speed and model uncertainties. This compensation is achieved without requiring real-time updates to the GP model or modifications to the feedforward control law, thereby significantly reducing the computational burden on the on-board processors. The proposed compensator solves an online optimization problem that captures the effect of uncertainties in terms of a pseudo-engine speed value. Subsequently, the pseudo-engine speed is provided to an offline generated LUT for computing an equivalent compensating control input. Thus, a compensating control is generated without the need of updating system model or, LUT in real-time. The stability analysis presented in this paper demonstrates that under the application of such compensating control, engine combustion phasing is steered to a desired value in finite-engine cycles. To evaluate the effectiveness of the proposed uncertainty compensator, it is applied to the challenging task of combustion phasing control in a multi-fuel CI engine. Simulation results demonstrate that the proposed control scheme successfully tackles such challenging scenarios and regulates the system output to achieve the desired combustion phasing. The main contributions made in this paper are summarized as follows:
\begin{enumerate}
    \item This work introduces a new uncertainty compensator that ensures robust tracking of desired combustion phasing despite model uncertainties, without requiring real-time updates to the GPR model. This significantly reduces the computational burden of the on-board processors.
    \item Model uncertainties are accounted for through an online optimization that computes a pseudo-engine speed, which is then fed to an offline-generated LUT for compensating control inputs thereby eliminating the need for real-time LUT updates.
    \item Through stability analysis, combustion phasing of the engine is guaranteed to converge to the desired value in finite engine cycles under the proposed approach. Its effectiveness is validated through simulations studies where compression ignition (CI) engine system is subjected to different fuel blends, variations in engine speed, and inherent model uncertainties in real-time. 
\end{enumerate}

The remainder of this paper is organized as follows: Section \ref{sec_two} provides the overall framework for the proposed modeling and control approach,
including uncertainty compensator which effectively tackles model uncertainty in real-time. 
Section \ref{sec_three} discusses the simulation results, and Section \ref{sec:six} provides concluding remarks and potential future research directions.

\section{Methodology} \label{sec_two}
This section presents the problem formulation for achieving the desired combustion phasing in a multi-fuel compression ignition (CI) engine. The formulation lays the foundation for the subsequent development of a data-driven control framework, which is introduced in the following section.

\subsection{Control Formulation}
This subsection defines the inputs and output considered in this study to characterize the combustion performance of the CI engine, which are listed as follows \eqref{eq1}:
\begin{align} \label{eq1}
    \begin{matrix}
        \text{\textbf{Control Inputs}}~(u): & \text{Main Injection Timing (MIT)},\\
        & \text{Glow Plug Power (GPP)}\\
        \text{\textbf{External Inputs}}~(Z): & \text{Fuel Cetane number (CN)},\\
        & \text{Engine Speed (RPM)}\\
        \text{\textbf{Output}}~(y): & \text{Crank Angle location at}\\
        & \text{50\% fuel mass burn (CA50)}
    \end{matrix}
\end{align}
Note that inputs to the system is categorized into two types: control inputs and external inputs. The key difference between the two types of inputs is that the external inputs are beyond the end user's control, as they define the operating conditions of the engine and have a significant impact on combustion performance. On the other hand, the control inputs can be changed appropriately by the user through actuators. These control inputs are used for achieving or maintaining the desired fuel combustion performance, which is defined by the common combustion metric of CA50. The CA50 represents the crank angle at which 50\% of the fuel mass gets burned, after top dead center (TDC). Through CA50, an estimate of combustion characteristics can be deduced in terms of pressure rise rate and fuel conversion efficiency. A deviation of CA50 from its desired location generally leads to detrimental effects on the engine performance, and hence the need for an appropriate control strategy is critical for achieving the desired performance throughout the engine combustion cycles. In order to achieve better combustion performance, an engine may have to resort to multiple injections. However, for this work, single injection case referred to as main injection has been considered. The CA50 is dependent on the control input, main injection timing (MIT) and thus it's adjusted for achieving the desired performance. The other control input considered in this work is the ignition assistant power. The particular ignition assistant considered is the glow plug and the control input is thus the glow plug power (GPP). It is to be noted that fuels which have higher cetane numbers would burn well under varied operating conditions. However, such is not the case with low cetane fuels, where even after adjusting the injection timings, it is difficult to achieve a good combustion. The glow plug actuator is used to create a hot spot, which would help in igniting the fuel-air mixture for low-cetane number fuels \cite{amezcua2022ignition}. Thus, this work utilizes glow plug to assist in achieving the desired combustion for a larger range of fuel cetane numbers. The engine map hence can be represented through its inputs and outputs as:
\begin{align}
    y:= CA50 = f(Z,u)
\end{align}
here the inputs to the unknown function $f$ is the control $u=[u1,u2]^{\top}$ and external inputs $Z=[Z_1,Z_2]^{\top}$ where $Z_1$, $Z_2$, $u_1$ and $u_2$ represents $CN$, $RPM$, $MIT$ and $GPP$ respectively. Whereas, the output $y$ is CA50. The overall objective is to control the CA50 to achieve desired combustion.

\subsection{Proposed Control Architecture}
The overall architecture of the data driven control strategy used for this work is shown in Fig. \ref{UC_BD}. The closed-loop architecture addresses the impact of model uncertainties on combustion control through deploying two key components: a feedforward controller and an uncertainty compensator. For a given fuel type and engine speed, a control input is derived from offline generated feedforward LUT in order to achieve desired combustion phasing. However, in the presence of model uncertainties or disturbances, there arises error between actual engine and a Gaussian Process based model developed using available data. As this error goes beyond a threshold value, the uncertainty compensator is activated. The uncertainty compensator then intervenes to compute a pseudo-engine speed. This pseudo-speed is then supplied to both the LUT and the system model: it informs the model to simulate control behavior at this adjusted speed, and simultaneously updates the LUT to generate a new control input. The compensated input aiming to eliminate the uncertainty, is ultimately applied to the real engine, which continues to operate at the measured engine speed. Thus, the proposed control strategy allows for real-time adaptation to uncertainties while preserving the benefits of model-based feedforward control. The detailed functioning of both the feedforward and uncertainty compensation modules is discussed in the subsequent subsection. 

\begin{figure*}[h!]
        \centering
        \includegraphics[scale=0.2]{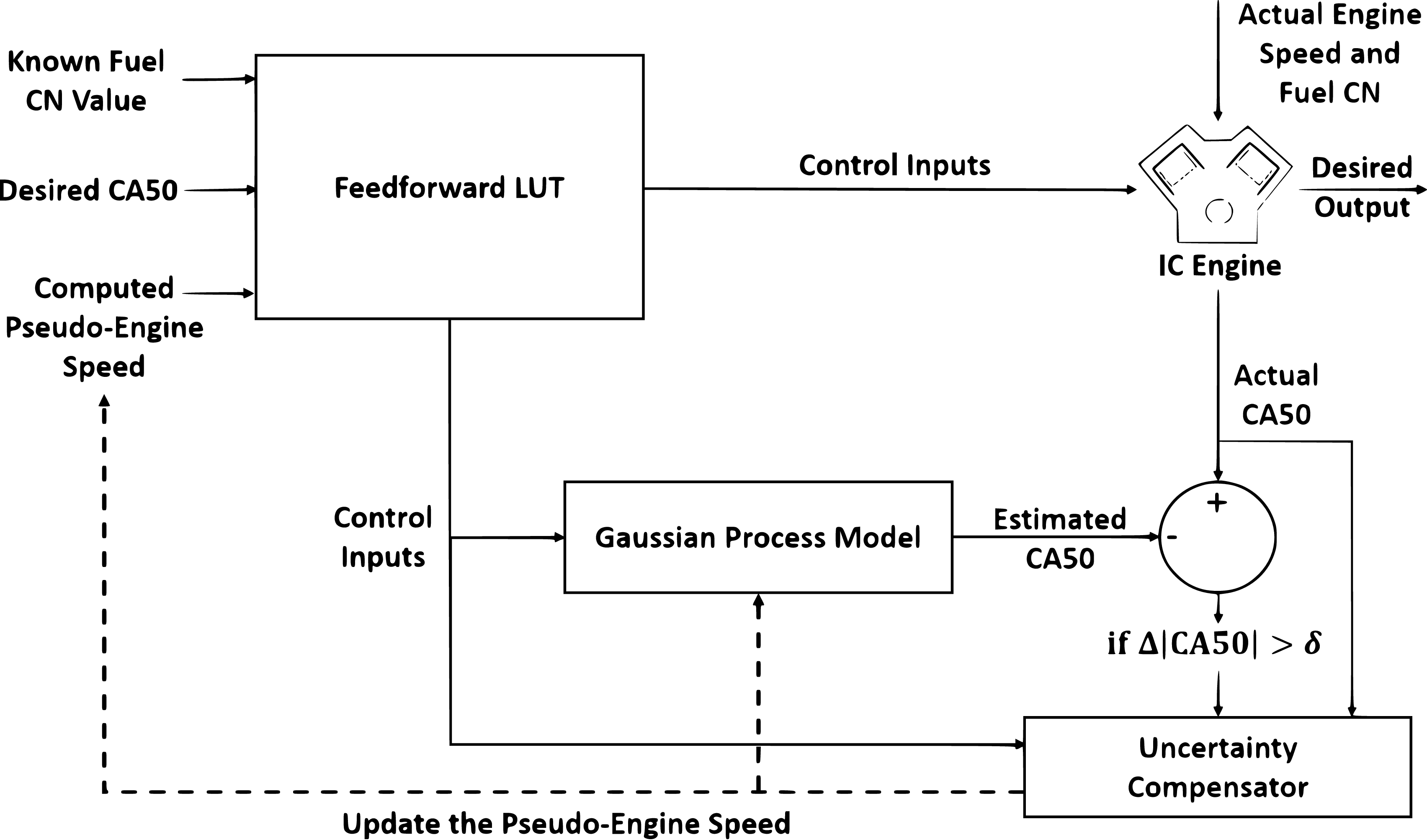}
        \caption{Block diagram of the overall architecture with Uncertainty Compensator}\label{UC_BD}
    \end{figure*}

\subsection{Feedforward Control Design}\label{sec:ControlDesign}
The feedforward control strategy is implemented in two stages: data collection and model inversion. In the first stage, data comprising control inputs $(MIT, GPP)$, external inputs representing operating conditions $(CN, RPM)$, and the corresponding output $CA50$ are obtained. This dataset is then used to train a surrogate model that captures the system behavior. In the second stage, model inversion is applied to the surrogate model to construct look-up tables (LUTs). These LUTs are designed such that, given a desired $CA50$ and measurable operating conditions $(CN, RPM)$, the corresponding control inputs $(MIT, GPP)$ can be generated to drive the engine's actual $CA50$ towards the desired value.  

\subsubsection{System modeling}

Data-driven modeling approaches are generally categorized into parametric and non-parametric methods. Parametric models assume a fixed structure and tune parameters using data, which may limit accuracy when system dynamics are not well known. This work employs Gaussian Process Regression (GPR), a non-parametric method that adapts flexibly to complex systems like combustion process. Once trained with suitable representative data, GPR which is a widely used probabilistic modeling method, predicts both the mean and uncertainty at any test point as
\begin{align}
    f \sim GP(m,k)
\end{align}
where, $f$ is a non-linear function that defines the gaussian process with a mean $m$ and covariance $k$ and this function captures the input–output behavior of the system under consideration. Considering practical systems are influenced by measurement noise $\epsilon$, gaussian process incorporates this as zero mean and variance $\sigma_{n}^{2}$, thereby representing output as:
\begin{align} \label{out_gpr}
    y = CA50(x)=f(x)+ \epsilon ~ \text{where}~ \epsilon \sim N(0,\sigma_{n}^{2}).
\end{align}
Here, the input is $x=[Z,u]\in \mathbb{R}^4$ and by combining the terms of the \eqref{out_gpr}, the output CA50 can also be represented as:
\begin{align}
CA50(x) \sim GP(m,k+\sigma_{n}^{2}\delta_{pq})
\end{align}
where, the Kronecker’s delta $\delta_{pq} = 1$ if $p=q$. The output GP has a chosen mean of $m=0$ and covariance $k+\sigma_{n}^{2}\delta_{pq}$ which are selected from varied options. Gaussian process modeling involves optimizing covariance function hyperparameters which determine how well the model fits the training data. To obtain the optimal hyperparameters, the log marginal likelihood function is maximized and considering that $N$ data-points fit the GP model, the output CA50 is now given as:
\begin{align}
CA50(X) \sim GP(m,K+\sigma_{n}^{2}I)
\end{align}
where, $X$ is the input vector, $m$ represents the vector consisting of mean values and the covariance matrix is given as $K=(k(x_i,x_j))~\forall i,j \in [1,N]$.

After fitting the model, it is important to perform predictions at any unknown input point. Let $CA50(X)$ be the known outputs, and $CA50(x')$ is the unknown output to be estimated at input $x'$ using the gaussian process model. Since the process is gaussian, the joint distribution of both known and unknown points can be written as:
\begin{align}
    \begin{bmatrix}
        CA50(X) \\
        CA50(x')
    \end{bmatrix} \sim N \left( \begin{bmatrix}
        m \\
        m(x')
    \end{bmatrix}, \begin{bmatrix}
        K+\sigma_{n}^{2}I) & K(X,x')\\
        K^T(X,x') & k(x',x')
    \end{bmatrix} \right).
\end{align}
Here, $CA50(X)$ is the training set, and $CA50(x')$ is the unknown output. $K(X,x')=k(x_i,x')~\forall i \in [1,N]$ represent the covariance
calculation between known data set $X$ and unknown test point $x'$.

Using the above expression, the conditional distribution of $CA50(x')$ given $CA50(X)$ can be obtained. For the sake of brevity, the mathematical discussion for the model development is eliminated and can be found in \cite{williams2006gaussian} for more details. After performing mathematical calculations, the model estimates at any
unknown point $x'$ is written as:
\begin{align}
\begin{matrix}
    CA50_m(x') = &m(x') + K(X,x')K^{-1}(CA50(X) - m) \\
    CA50_{\sigma^2}(x') = &k(x',x') - K(X,x')^TK^{-1}K(X,x')
\end{matrix}
\end{align}
where, $CA50_m(x')$ are the mean and $CA50_{\sigma^2}(x')$ variance estimate from the gaussian process at any unknown point.

\subsubsection{Optimal look-up table generation}

Once the surrogate model has been developed, the next step involves the control design to ensure desirable tracking. The control design in this work involves model inversion in order to obtain a feedforward control. This model inversion is achieved through formulation of an optimization problem and the cost function chosen to achieve this objective is given as:
\begin{align}
    J_i(x)= & |CA50^i_{desired} -CA50_m(x)|\\ \nonumber
    & +p_1 \times CA50_{\sigma}(x)+p_2 \times GPP
\end{align}
where, the $i^{th}$ setpoint of CA50 is given as $CA50_{desired}^{i}$ with $i=1,2,..., R$. The mean $CA50_m(x)$, and uncertainty $CA50_{\sigma}(x)$, estimates of the surrogate model is evaluated at $x= [MIT,GPP,CN_j,RPM_k]$. Here,  $CN_j$ represents the $j^{th}$ known fuel blend and $RPM_k$ represents $k^{th}$ known engine speed where $j,k = 1, 2, ...$. Also, $p_{1}$ and $p_2$ are the chosen weights for the model uncertainty and glow plug power. 

The cost function $J_i(x)$ is  formulated to identify the optimal control input $x$ that drives the predicted combustion phasing $CA50_m(x)$ towards a target value of $CA50_{desired}$, while simultaneously penalizing model uncertainty and the usage of GPP. To discourage the optimizer from selecting inputs in regions where the model exhibits high predictive uncertainty, a penalty term weighted by the constant $p_1$ is introduced. This ensures the solution remains within well-characterized regions of the model. An additional term, scaled by $p_2$, penalizes excessive reliance on GPP. More details can be found in \cite{dong2022data,pal2024data}. The overall optimization problem can be formulated as:
\begin{align}
\label{gen_opt_prob}
\begin{matrix}
\underset{x}{\text{min}} & J_i(x)\\
\text{Subject to} &  ~~~~~~~MIT_{prev} \leq MIT \leq MIT_{prev} +5\\
&  ~~~~~~~GPP_{prev} \leq GPP \leq GPP_{ub} \\
\end{matrix}
\end{align}
where $MIT_{prev}$, $GPP_{prev}$ and $GPP_{ub}$ stands for admissible control input range and have been discussed in details in \cite{pal2024data}.
Here, the input vector is defined as $x= [MIT,GPP,CN,RPM]$, where $(MIT, GPP)$ represent the control variables, and $(CN$, $RPM)$, denote the current operating conditions. The optimal control input $x$ is computed by solving the constrained optimization problem formulated in \eqref{gen_opt_prob}, given as:
\begin{align}
    x^* = \underset{x}{argmin}~J(x)
\end{align}
The optimal solution $x^*$ steers the system model to the desired value for a given $Z =[CN,RPM]^{\top}$. Note that it may not be always feasible to solve the optimization problem in real-time due to huge computational requirement. One way to avoid such online computation can be to generate a feedforward lookup table (LUT) which requires repeatedly solving the optimization process offline for various desired values. Each entry of the lookup table (LUT) defines the control vector
$u = [\mathrm{MIT}, \mathrm{GPP}]^{\top}$
required to achieve a prescribed combustion phasing target 
$y_d = \mathrm{CA50}_{desired} \in \mathbb{R}$
at a specified operating condition $ Z $. Consequently, for any admissible pair $ (y_d, Z) $, the control input $ u $ is obtained by interpolation over the discrete LUT, thereby enabling model-based output regulation. Let the operating condition be parameterized as
$Z = [Z_1, Z_2]^{\top}$, where $ Z_1 $ denotes fuel cetane number and $ Z_2 $ denotes the physical engine speed. The Gaussian Process (GP) surrogate model defines the input–output mapping
\begin{align}
y_m &= f_m(Z,u),
\end{align}
with the LUT-generated control expressed as $u = g(y_d,Z)$. For notational simplicity, we denote $g(Z) \equiv g(y_d, Z)$ throughout the remainder of the paper. By construction of the LUT via model inversion, the composite mapping satisfies
\begin{align}
f_m\bigl(Z, g(Z)\bigr) &= y_d,
\label{eq:fh-LUT-yd}
\end{align}
i.e., the surrogate model output coincides with the desired combustion phasing at the specified operating condition. To characterize the corresponding closed-loop behavior under the true plant dynamics $ f(\cdot) $, we define the composite real map induced by the LUT as
\begin{align}
\Phi(Z) &:= f\bigl(Z, g(Z)\bigr).
\label{eq:Phi-def}
\end{align}
\begin{remark}
    It is important to emphasize that while the LUT performs interpolation to generate the control input based on $y_d$ and $Z$, the inclusion of additional constraints ensures that the LUT consistently produces a bounded control input. That is, for any $y_d$ and $Z$, the LUT generates admissible control input $u$.
\end{remark}

\subsection{Real-time Uncertainty Compensation}
In the case of CI engines, collecting data across a wide range of operating conditions and control inputs is both costly and time-intensive. As a result, the limited availability of a sufficiently large and diverse dataset can reduce the accuracy of the surrogate model. Consequently, the control inputs obtained from LUT, for a given desired $CA50$ and operating conditions $CN$ and $RPM$, may fail to achieve the desired engine performance. This necessitates the design of an additional compensating control mechanism to handle such scenarios. To this end, the present work proposes a novel real-time uncertainty compensator algorithm where instead of giving real engine speed (RPM) $Z_2^0$ as an input to LUT, a compensating pseudo-RPM $Z_2$ is provided to LUT. Thus, the composite real map \eqref{eq:Phi-def} induced by the LUT becomes
\begin{align}
\Phi(Z) &:= f\bigl(Z^0, u\bigr)= f\bigl(Z^0, g(Z)\bigr)
\end{align}
where $ Z^0= [Z_1^0, Z_2^0]^{\top} $ denotes the actual measured operating condition in which real engine operates while $Z = [Z_1^0, Z_2]^{\top}$. This mapping quantifies the model–plant discrepancy along the inversion manifold defined by the LUT and forms the basis for subsequent compensating mechanism design.

\subsubsection{Uncertainty Compensator Algorithm} 

Considering normal engine operation, where fuel cetane $CN$ and engine speed $RPM$ are known and the optimal control strategy is trying to achieve a fixed
$CA50$ setpoint, $CA50_{desired}$ i.e. $y_d$, then the performance of the control strategy can be defined by the error $e$ as:
\begin{align}\label{eqn:err}
    e &= y - f_m(Z^{cur},u^{cur}) = y - f_m\bigl(Z^{cur},g(Z^{cur})\bigr)
\end{align}
where $y$ represents the actual CA50 value computed from in-cylinder pressure signal in real-time and $f_m(Z^{cur},u^{cur})$ is the
Gaussian model mean estimate at optimal control inputs $u^{cur} = g(Z^{cur})$ generated from LUT with $Z^{cur}=[Z_1^0,Z_2^{cur}]$ where $Z_1^0$ is measured $CN$ number and $Z_2^{cur}$ is engine speed. Ideally, the error $e = 0$ defines perfect control performance. However, this is impossible to achieve in actual engine operation due to the stochastic nature of engine operation. Therefore, at optimal control setting, the error $e$ is bounded by a small value $\epsilon$ i.e. $|e|<\epsilon$.

Now, assuming the the real engine runs at operating conditions or under the application of control inputs for which the model has not been adequately trained. In such cases, model uncertainty arises, leading to a deviation in $y$, while $y_d$ remains unchanged. This occurs because the feedforward controller lacks the ability to adapt its control inputs to account for these uncertainties. This will cause the error $e$ to increase and thus demand a compensating control to minimize the error. In this work, the compensating control is generated through finding an equivalent speed condition $\hat{Z}_2^{new}$ that reduces the error  $e^{new}$ defined as follows:
\begin{align}\label{eqn:err2}
    e^{new} &= y - f_m\bigl(\hat{Z}^{new},g(\hat{Z}^{new})\bigr)
\end{align}
where $\hat{Z}^{new}=[Z_{1}^0,\hat{Z}_{2}^{new}]$ and $\hat{Z}_{2}^{new}$ represents the engine speed at which, if the surrogate model is assumed to operate, the predicted output $y_m\bigl(\hat{Z}^{new},g(\hat{Z}^{new})\bigr)$ closely aligns with the actual output $y$ of the real engine, thereby minimizing the error $e^{new}$. The value $\hat{Z}_2^{new}$ can be obtained by solving an optimization problem which is defined as follows:
\begin{align}\label{opt_prob}
&\underset{\hat{Z}_2^{new}\in Z_2^*}{min}~\left|\left(y+\gamma^{new}\right)- f_m\bigl(\hat{Z}^{new},g(Z^{cur})\bigr)\right| 
\nonumber\\ &~~~~~~~~~~~~~~~~~~~~~~~~~~~~~ 
+\beta \left|\hat{Z}_2^{new}-Z_2^{cur}\right|
        \end{align}
where $Z_2^*$ represents the set of all operating points that the optimization algorithm can visit while ensuring convergence is achieved gradually over engine cycles. The weight $\beta$ is a small positive scalar that facilitates choosing solution as close as possible to $Z_2^{cur}$. The control inputs generated by $g(Z^{cur})$ are optimal $MIT$ and $GPP$ obtained from LUT corresponding to pseudo-speed $Z_2^{cur}$, measured fuel cetane $CN$ and desired CA50 $y_d$. The parameter $\gamma$ in \eqref{opt_prob} is referred in this work as an adaptive bias and is designed to evolve based on two conditions which are discussed next. For that first, we represent two expressions as follows:
\begin{align*}
    T_1:~&\bigl(y-y^{prev}\bigr) \\ T_2:~&\bigl(y+y^{prev}-2y_d\bigr) 
\end{align*}
where $y^{prev}$ stands for the previous value of $CA50$. Using above two conditions, $\gamma$ is designed as follows:
\vspace{0.2cm}\\
\noindent
If $T_1*T_2<0$, then
\begin{subequations}\label{eqn:gamma}
\begin{equation}\label{gam_cond1}
    \gamma \leftarrow 0,
\end{equation}
otherwise, 
\begin{equation}\label{gam_cond2}
    \gamma \;\leftarrow\; \min\bigl(|\gamma+\Gamma \text{sign}(T_1(k))|,\gamma_{\max}\bigr)\text{sign}\bigl(\gamma+\Gamma\text{sign}(T_1(k))\bigr)
\end{equation}
\end{subequations}
The parameter $\Gamma>0$ is a user-defined small positive scalar and \textit{sign} stands for signum function which is defined as:
\begin{align*}
    \text{sign(x)} = \left\{\begin{matrix}
        +1 & \text{if}~x>0 \\
        0 & \text{if}~x=0 \\
        -1 & \text{if}~x<0 \\
    \end{matrix}\right.
\end{align*}

Note that since the solution of optimization problem \ref{opt_prob} is not the real engine speed, $\hat{RPM}$ or $\hat{Z}_2$ is referred in this work as the \textit{pseudo-engine speed} or \textit{pseudo-RPM}. On achieving pseudo-RPM through solving optimization problem \eqref{opt_prob}, it is provided to the LUT to generate a compensating control that reduces error $e^{new}$. It is also important to note here that the pseudo-RPM does not possess any inherent physical significance and should not be confused with the real measured value of engine speed. 

\begin{remark}
    The proposed control strategy shares a conceptual similarity with Model Reference Adaptive Control (MRAC) in that both aim to handle model uncertainty and ensure the system output tracks a desired reference. However, the key difference lies in how this is achieved. Traditional MRAC requires a known model structure and adapts controller parameters in real time to minimize the tracking error between the system and a reference model. In contrast, the present method uses a data-driven approach based on the trained GP model, which may not be accurate at all operating conditions. Rather than adapting controller parameters online, it uses a separate compensator to adjust a pseudo-speed value used by the LUT to switch between different feedforward control maps. Further, this also avoids the need for real-time adaptation of the model and controller (i.e. retraining the model and updating the feedforward control through LUT update). On the other hand, unlike MRAC where the plant and model typically receive different inputs during adaptation, here both the model and the real system are given the same control input always. Since the LUT and the model are derived from the same trained mapping, providing the same pseudo-RPM to the LUT and operating the model at that pseudo-RPM ensures that the model consistently generates the desired output. To summarize, the proposed approach ensures performance without requiring structural knowledge of the system or online control parameter adaptation, distinguishing itself clearly from classical MRAC. 
\end{remark}

\begin{algorithm}[H]
\caption{Uncertainty Compensator}
\label{alg:corrected}
\begin{algorithmic}[1]
\Require Desired CA50 $\yd$, measured $Z^0=[Z_1^0,Z_2^0]$, GP model~$f_m(Z_1,Z_2)$, LUT~$g(Z_1,Z_2)$, convergence threshold $\varepsilon_{\mathrm{thr}}>0$ and minimum step $\eta_{\min}$ with $0<\eta_{\min}\le\varepsilon_{\mathrm{thr}}/\mu_{\max}$.
\While{$|y^{prev}-\yd|\ge\varepsilon_{\mathrm{thr}}$} 
\Statex \hspace{\algorithmicindent}\rule{0.85\linewidth}{0.4pt}
\Statex \hspace{\algorithmicindent}\textsc{Step 1: Apply current pseudo-$Z$ to plant.}
\State Compute feedforward command $u^{cur}=g(Z^{\mathrm{cur}})$.
\State Apply $u^{cur}$ to the real engine (at fixed $Z^0$); 
\State Measure CA50:~$y^{cur}$ and $e(k)$
\Statex \hspace{\algorithmicindent}\rule{0.85\linewidth}{0.4pt}
\Statex \hspace{\algorithmicindent}\textsc{Step 2: Progress check and $\gamma$ update.}
\State $T_1(k)\leftarrow y^{cur}-y^{prev}$
\State $T_2(k)\leftarrow y^{cur}+y^{prev}-2\yd$
\If{$T_1\,T_2<0$} 
\State $Z^{\mathrm{prev}}\leftarrow Z^{\mathrm{cur}}$ 
\State $y^{prev}\leftarrow y^{cur}$
\State $\gamma\leftarrow 0$ 
\Else 
\State Update $\gamma$ using \eqref{gam_cond2}
\EndIf
\Statex \hspace{\algorithmicindent}\rule{0.85\linewidth}{0.4pt}
\Statex \hspace{\algorithmicindent}\textsc{Step 3: Model-based optimization}
\State Solve \eqref{opt_prob} with $\gamma$  to obtain $\hat{Z}_2^{new}$
\Statex \hspace{\algorithmicindent}\rule{0.85\linewidth}{0.4pt}
\Statex \hspace{\algorithmicindent}\textsc{Step 4: Perturbation safeguard.}
\State $\Delta Z_2\leftarrow \hat{Z}_2^{new}-Z_2^{cur}$ 
\State $Z_2^*\leftarrow\min\!\bigl(\delta_0,\,|y^{prev}-\yd|/\mu_{\max},\,R\bigr)$ 
\If{$|\Delta Z_2|>Z_2^*$} 
\Statex \hspace{\algorithmicindent}\emph{Phase 1: $\gamma$-bisection (model direction)}
\State $\gamma_{\mathrm{trial}}\leftarrow\gamma$
\For{$j=1,\ldots,J_{\max}$}
    \State $\gamma_{\mathrm{trial}}\leftarrow\tau_s\cdot\gamma_{\mathrm{trial}}$
    \State Recompute $\hat{Z}_2^{new}$ solving \eqref{opt_prob} with $\gamma_{\mathrm{trial}}$ 
    \State $\Delta Z_2\leftarrow\hat{Z}_2^{new}-Z_2^{cur}$
    \If{$|\Delta Z_2|\le Z_2^*$}
        \State $\gamma\leftarrow\gamma_{\mathrm{trial}}$
        \State \textbf{go to} Step~5
    \EndIf
\EndFor
\Statex \hspace{\algorithmicindent}\emph{Phase~2: Clamping fallback (guarantees $|\Delta Z_2|\le Z_2^*$).}
\Statex $\hat{Z}_2^{new}\leftarrow Z_2^{cur}+Z_2^*\cdot\text{sign}(\Delta Z_2)$ 
\State $\gamma\leftarrow 0$ 
\EndIf
\State $\Delta Z_2\leftarrow \hat{Z}_2^{new}-Z_2^{cur}$ 
\State $\eta_{\mathrm{eff}}\leftarrow\min(\eta_{\min},\,Z_2^*)$
\If{$0<|\Delta Z_2|<\eta_{\mathrm{eff}}$}
  \State $\hat{Z}_2^{new}\leftarrow Z_2^{cur}+\eta_{\mathrm{eff}}\cdot\text{sign}(\Delta Z_2)$ (min.-step floor)
\EndIf
\Statex \hspace{\algorithmicindent}\rule{0.85\linewidth}{0.4pt}
\Statex \hspace{\algorithmicindent}\textsc{Step 5: Update current pseudo-coordinate.}
\State $Z^{\mathrm{cur}}\leftarrow[Z_1^0,\,\hat{Z}_2^{new}]$
\EndWhile
\end{algorithmic}
\end{algorithm}

Algorithm \ref{alg:corrected} summarizes the uncertainty compensator with the perturbation safeguard step. The algorithm becomes active when the CA50 tracking error $e$ exceeds a predefined threshold. In this work, the threshold is chosen based on the experimentally observed cycle-to-cycle variation of CA50 across the considered fuel blends and engine speeds. At each engine cycle, an optimization problem \eqref{opt_prob} is solved using the GP model to produce a candidate pseudo-RPM $\hat{Z}_2^{new}$ that is expected to reduce the CA50 error when queried through the LUT. The resulting effective perturbation $\Delta Z_2$ defined in line 15 of Algorithm \ref{alg:corrected}, is then passed through a safeguard step (Step 4) that shrinks its magnitude, via bisection. In the next step, the perturbation is clamped (if needed) so that $|\Delta Z_2| \leq Z_2^*$ where $Z_2^*$ defined as in line 16, ensuring convergence. While on the other hand, Zeno behavior is avoided through using $\eta_{\text{eff}}$ as in Line 33 (minimum step floor). The parameters that characterize $Z_2^*$ is discussed in Remark \ref{conv_param}. During each cycle, the safeguarded perturbation is evaluated using the plant-based acceptance condition $T_1 *T_2 \leq 0$. If the condition is satisfied, the current CA50 and the associated pseudo-RPM are stored as the updated pair $(y^{\mathrm{cur}}_\ell, Z^{\mathrm{cur}}_{2,\ell})$, referred to as \textit{anchor}, where $\ell$ denotes the anchor index, incremented upon each successful acceptance. If the acceptance condition is not satisfied, the algorithm retains the previously accepted pair $(y^{\mathrm{cur}}_\ell, Z^{\mathrm{cur}}_{2,\ell})$ and updates the parameter $\gamma$ according to \eqref{gam_cond2}. This initiates an inner loop in which $\gamma$ is iteratively adapted over successive engine cycles, starting from $\gamma = 0$, and used within the optimization problem to generate candidate pseudo-RPM values $\hat{Z}_2^{\mathrm{new}}$. The inner loop continues until a value of $\gamma$ is found such that the resulting CA50 satisfies the acceptance condition $T_1 * T_2 \leq 0$. Once a valid update is obtained, the inner loop terminates (guaranteed through Theorem \ref{thm:convergence}), $\gamma$ is reset to zero, and the accepted pair $(y^{\mathrm{cur}}_\ell, Z^{\mathrm{cur}}_{2,\ell})$ is updated. The algorithm then proceeds to the next outer-loop iteration. This entire process is repeated until the CA50 tracking error falls below the prescribed tolerance $\varepsilon_{thr}$.

\begin{remark}\label{conv_param}
    The parameters in Algorithm \ref{alg:corrected}, $\delta_0$, $R$, and $|y^{prev}-\yd|/\mu_{\max}$ define the favorable step-size region in which the plant--LUT composition behaves predictably and the GP-based correction moves CA50 in the correct direction. The constant $\delta_0 > 0$ is the largest neighborhood around zero in which any effective perturbation $\Delta Z_2$ generated by the optimization produces a CA50 change (i.e. $\Delta y$) with sign opposite to the current anchor error, and $R > 0$ is the radius of the ball in pseudo-RPM space over which the sector bounds
\begin{align}
\mu_{\min}\,|\Delta Z_2| \leq |\Delta y(\Delta Z_2)| \leq \mu_{\max}\,|\Delta Z_2|
\end{align}
hold uniformly for the plant--LUT composition. The constants $\mu_{\min}$ and $\mu_{\max}$ represent the corresponding lower and upper sector gains over this region. Because the stopping criterion enforces $|e_\ell|\ge\varepsilon_{\mathrm{thr}}$ throughout operation, the favorable threshold satisfies $Z_2^*=\min(\delta_0,\,|y^{prev}-\yd|/\mu_{\max},\,R)\ge\min(\delta_0,\,\varepsilon_{\mathrm{thr}}/\mu_{\max},\,R)>0$, so the region never degenerates as the error shrinks. The applied perturbation is additionally lower-bounded by $\eta_{\min}$ in Step~4. Therefore, each accepted step satisfies $\eta_{\min} \leq \left|\Delta Z_2^{(\ell)}\right| \leq Z_2^*$, where $\eta_{\min} \leq \varepsilon_{\mathrm{thr}}/ \mu_{\max}$ ensures that no overshoot is introduced. All of these parameters involved in characterizing $Z_2^*$ must be selected in accordance with the prescribed construction to ensure the convergence guarantees of Theorem \ref{thm:convergence}.
\end{remark}

\subsubsection{Stability Analysis}\label{sec:stability}
In this subsection, a theorem is provided to guarantee that under the proposed control architecture, $CA50$ converges to $CA50_{desired}$ over the engine cycles thereby eliminating the effect of uncertainties adversely affecting the system performance.

\begin{theorem}[Convergence of Proposed Controller]\label{thm:convergence}
Under Assumptions B1--B5 (mentioned in Appendix~\ref{appendixB}) with convergence threshold $\varepsilon_{\mathrm{thr}}>0$, and assuming $|e_0|\ge\varepsilon_{\mathrm{thr}}$, the following hold for Algorithm~\ref{alg:corrected}, where $\ell=0,1,2,\dots$ indexes the accepted anchors, $e_{\ell}=y^{prev}_{\ell}-\yd$, and $V_{\ell}=1/2~e_{\ell}^{2}$.
\begin{enumerate}
\item [(i)] \textit{Finite inner-loop termination:} For every anchor~$\ell$ with $|e_{\ell}|\ge\varepsilon_{\mathrm{thr}}$, there exists a finite number of iterations such that the acceptance condition
$T_1 * T_2 < 0,$
evaluated using $(y^{\mathrm{cur}},\,y^{\mathrm{prev}})$, is satisfied, and the anchor is updated, i.e., $(Z^{\mathrm{prev}}, y^{\mathrm{prev}}) \leftarrow (Z^{\mathrm{cur}}, y^{\mathrm{cur}})$.
\item [(ii)] \textit{Strict monotone decrease:} The Lyapunov function $V_{\ell}=1/2~e_{\ell}^{2}$ satisfies
\begin{align}
V_{\ell+1} < V_{\ell}, \quad \forall~\ell \text{ with } |e_{\ell}|\ge\varepsilon_{\mathrm{thr}}.
\end{align}
\item [(iii)] \textit{Finite-time practical convergence:} The algorithm terminates after a finite number of outer-loop steps $\ell^{*}<\infty$, at which point $|e_{\ell^{*}}|<\varepsilon_{\mathrm{thr}}$. Moreover,
\begin{align}
\ell^{*}\;\le\;\Bigl\lfloor \frac{V_0}{\delta(\varepsilon_{\mathrm{thr}})}\Bigr\rfloor+1,
\end{align}
where $\delta(\varepsilon_{\mathrm{thr}})=\tfrac{1}{2}\,\mu_{\min}\,\eta(\varepsilon_{\mathrm{thr}})\bigl(2\varepsilon_{\mathrm{thr}}-\mu_{\min}\,\eta(\varepsilon_{\mathrm{thr}})\bigr)>0$ is the minimum Lyapunov decrease per outer-loop step when $|e_{\ell}|\ge\varepsilon_{\mathrm{thr}}$.
\end{enumerate}
\end{theorem}
\begin{proof}
    Refer to proof provided in Appendix \ref{thm_proof}.
\end{proof}

\section{Results and Discussion}\label{sec_three}
This section presents the results obtained from conducting simulation studies and demonstrates how the engine performs under the application of the proposed data-driven modeling and control framework. 

\subsection{Data Collection}
For the simulation studies presented in this work, the dataset is generated using Computational Fluid Dynamics (CFD) simulations conducted at Computational Reactive Flow \& Energy Lab (CRFEL) of University of Minnesota–Twin Cities. These simulations are based on a GM 1.9 L single-cylinder EACI metal engine. Key engine specifications are summarized in Table \ref{tb:engine}, and further details are available in \cite{narayanan2024simulation}.
\begin{table}[h!]
\centering
\caption{Key Specifications of the EACI Metal Engine}
\begin{tabular}{l c}
\hline
\textbf{Parameters} & \textbf{Values} \\
\hline
Bore & 83 mm \\
Stroke & 90.4 mm \\
Connecting Rod Length (CR Length) & 144.8 mm \\
Displacement & 0.48 L \\
Effective Compression Ratio & 15.1 \\
Swirl Ratio & 1 \\
Inlet Valve Closing (IVC) & -132\textdegree{} CA \\
Number of Nozzle Holes & 7 \\
Nozzle Diameter & 141 $\mu$m \\
Spray Cone Angle & 18\textdegree \\
Spray Tilt Angle & 70\textdegree \\
\hline
\end{tabular}
\label{tb:engine}
\end{table}
The data is collected under steady-state engine operation across a range of fuel cetane numbers (CN), engine speeds (RPM), and control inputs namely, Glow Plug Power (GPP) and Main Injection Timing (MIT). Specifically, operating conditions are sampled at $CN = [35, 42, 48]$ and $RPM = [1200, 2000, 2800]$. For each condition, control inputs are varied over $MIT \in [-25, -3]$ and $GPP \in [0, 70]$ to obtain the corresponding combustion phasing output, $CA50$. Through this process a dataset of 217 points comprising input-output tuples of the form (MIT, GPP, CN, RPM, CA50) were generated.

\subsection{Modeling and Feedforward Control Design}
One of our previous works \cite{dong2022data} proposed the design of a data-driven feedforward control framework tailored specifically for a fixed fuel type. In this approach, control inputs were generated based on a surrogate model that was valid only for a particular fuel property, thereby limiting its applicability across fuels with varying characteristics. Another of our studies \cite{pal2024data} addressed this limitation by presenting a method to synthesize unified feedforward control policies for engines operating with different fuel cetane numbers (CN). This extension allowed for generalization over fuel variability, making the control strategy more versatile across a broader set of real-world conditions. In addition to CN as operating condition as our previous works \cite{pal2024data,govindraju2023rate}, this work incorporates the engine speed as an additional operating condition thereby bringing in further, complexity in combustion behavior dynamics. Based on the collected
data, the Gaussian process model was created as discussed in \cite{dong2022data,williams2006gaussian}. For simulations, the real engine is considered to be a GP model trained on the full dataset of 217 points. While, the surrogate model is trained on a subset of 32 data points, selected from the full dataset, with base operating conditions fixed at $CN = [35, 42, 48]$ and $RPM = [1200, 2800]$. Subsequently, the surrogate model is utilized to construct feedforward LUT through model inversion by solving the optimization problem \cite{dong2022data}. In this work, the weighting factors in the optimization problem are selected as $p_1=1$ and $p_2=0.01$ to balance tracking performance with robustness and actuation efficiency. To ensure that the control inputs derived from the LUT remain within permissible bounds, the actuator limits are imposed as constraints in the optimization framework \cite{govindraju2023rate}. To perform optimization, the ga function from MATLAB Global Optimization Toolbox is used to implement the genetic algorithm with the population size = 100 and total iteration number = 500.

\begin{figure}[h!]
    \centering
    \includegraphics[width=0.55\linewidth,height = 0.4\linewidth]{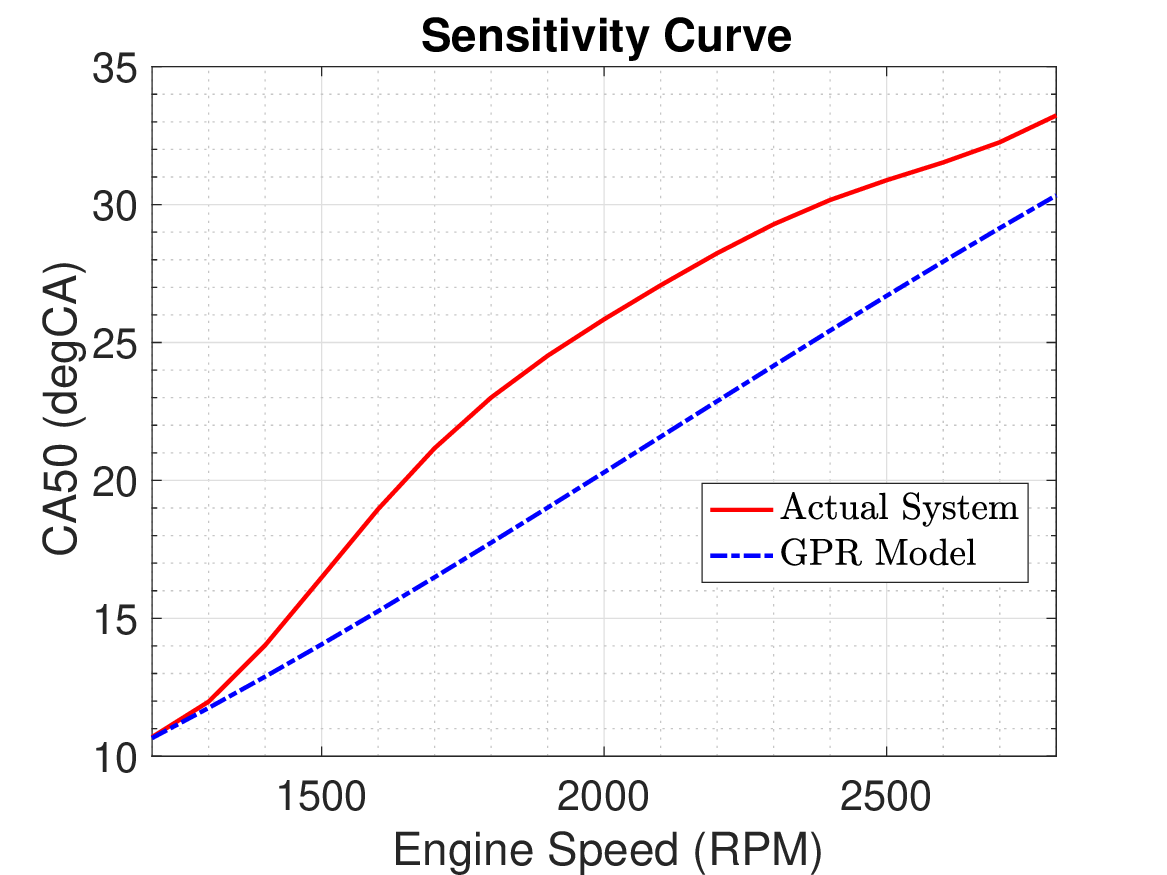}
    \caption{Variation of $CA50$ with engine speed for $CN=36$, $MIT=-7$ deg CA and $GPP=0$ W}
    \label{fig:sc}
\end{figure}

\begin{figure}[htp]
\centering

\begin{subfigure}{0.65\textwidth}
  \centering
  \includegraphics[width=0.85\linewidth,height = 0.55 \linewidth]{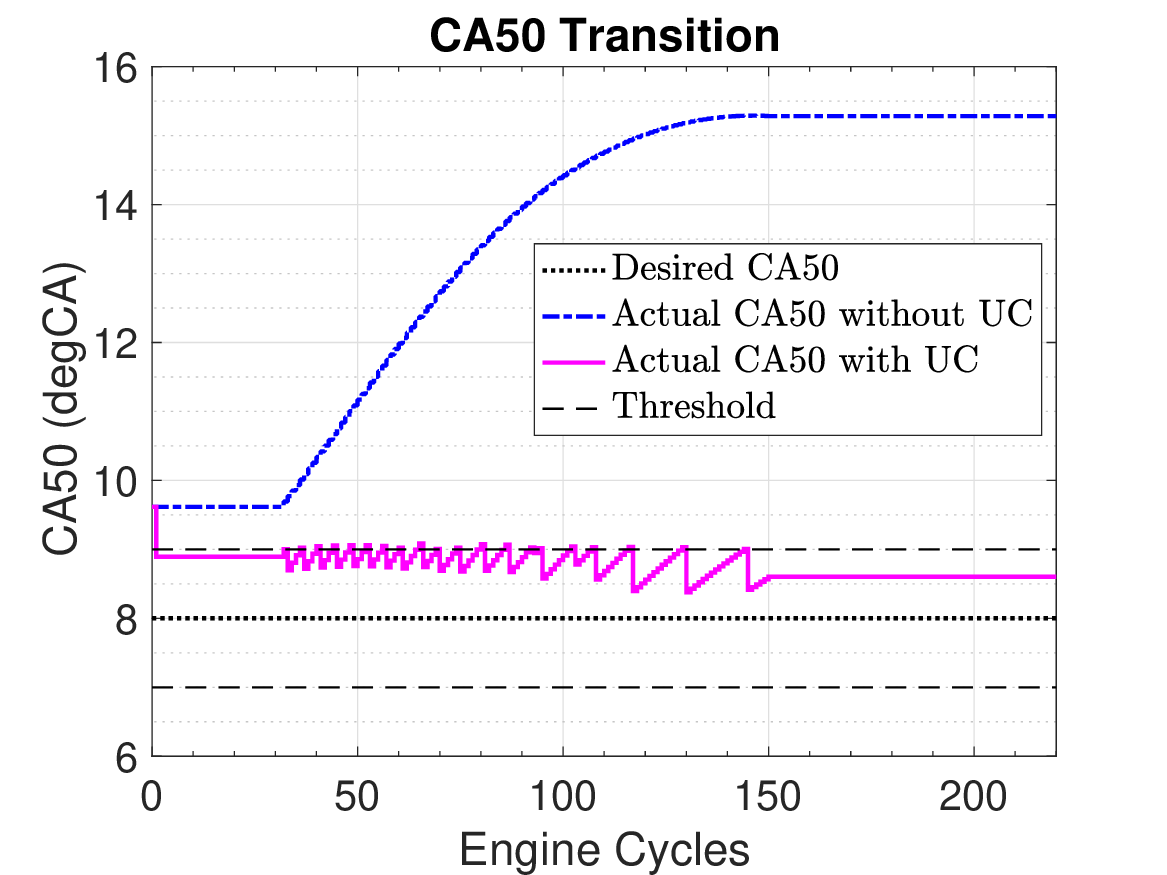}
\end{subfigure}


\begin{subfigure}{0.65\textwidth}
  \centering
  \includegraphics[width=0.85\linewidth,height = 0.55\linewidth]{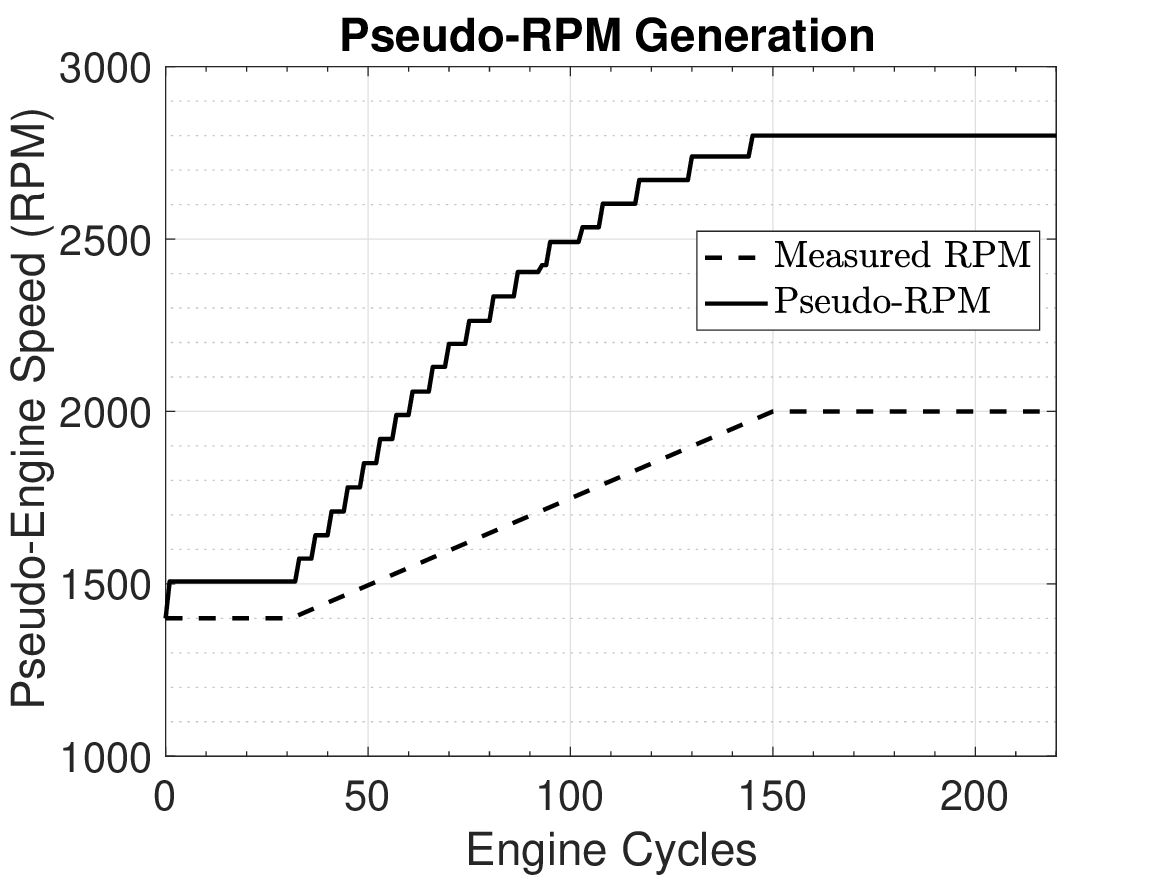}
\end{subfigure}


\begin{subfigure}{0.65\textwidth}
  \centering
  \includegraphics[width=0.85\linewidth,height = 0.55\linewidth]{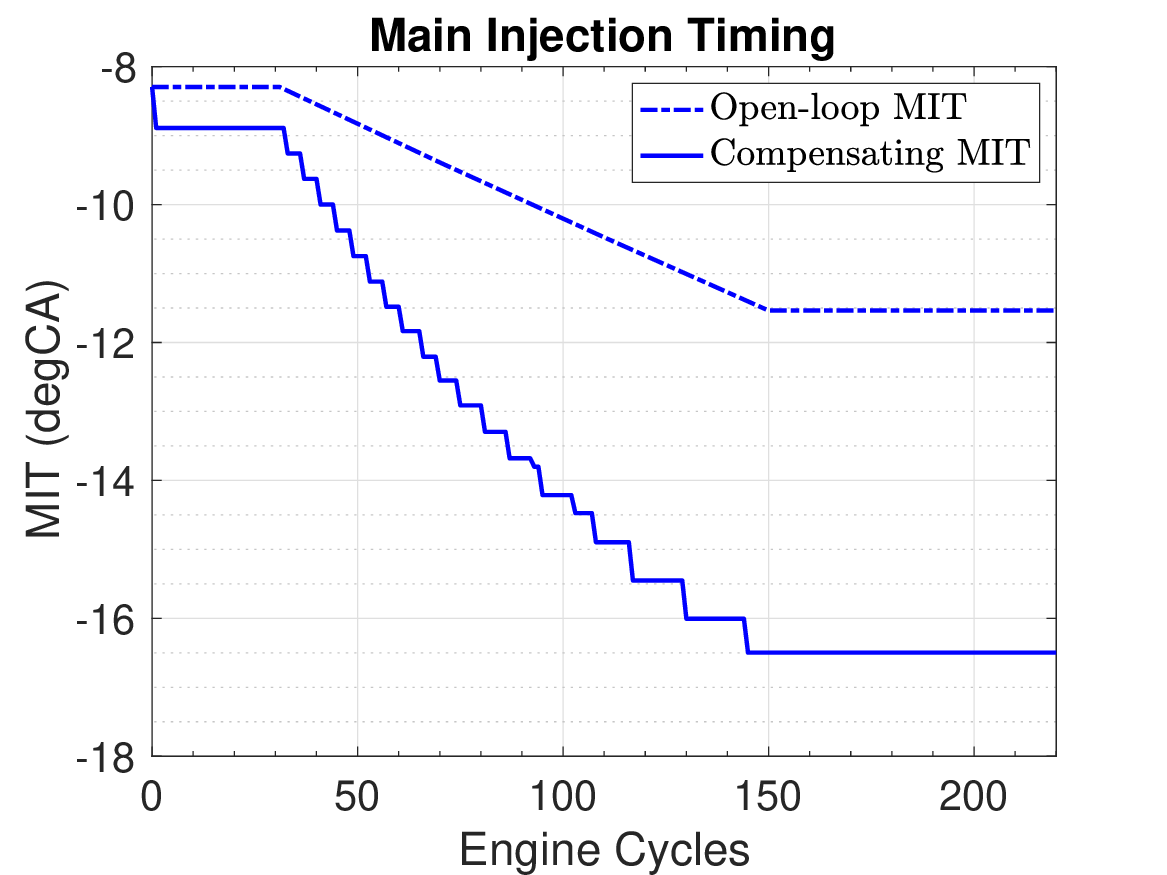}
\end{subfigure}


\begin{subfigure}{0.65\textwidth}
  \centering
  \includegraphics[width=0.85\linewidth,height = 0.55\linewidth]{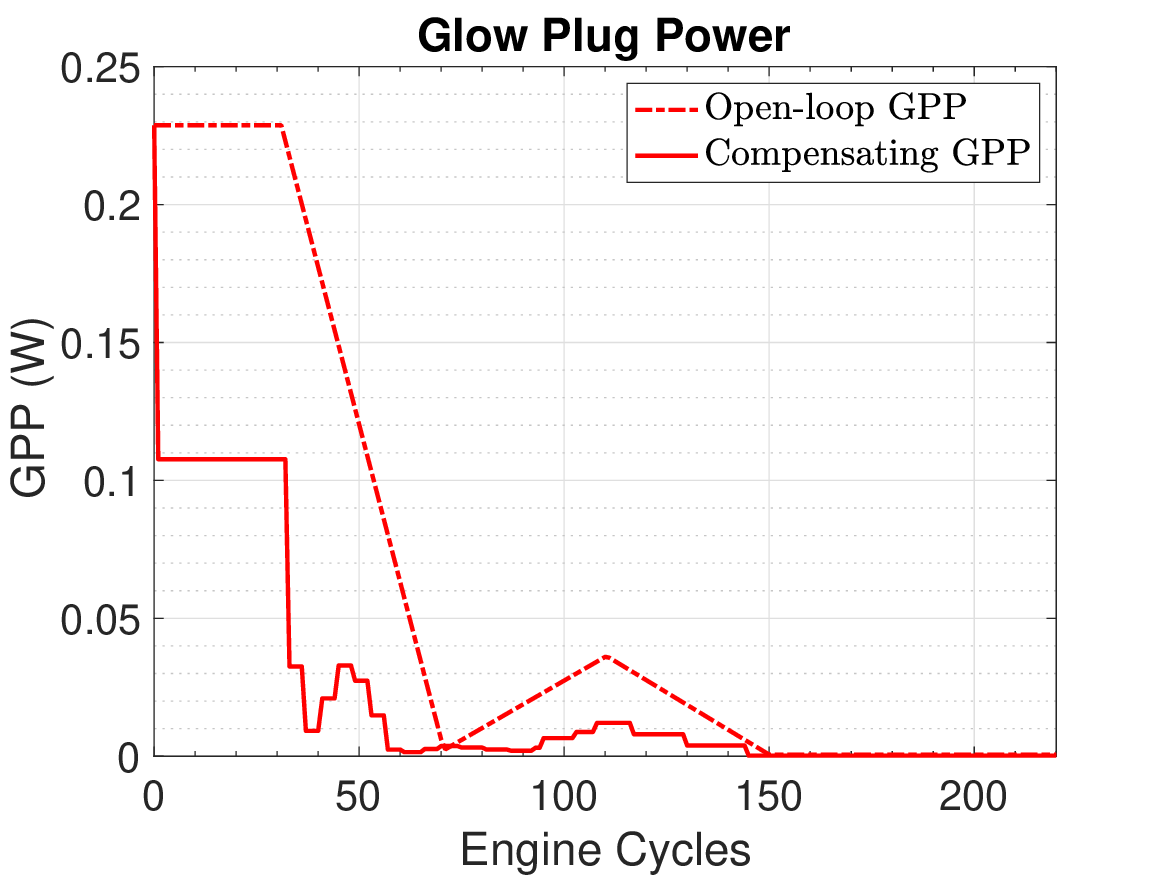}
\end{subfigure}

\caption{Simulation Results for Case 1 where CI engine is supplied with fuel $CN=45$ and engine speed changing from $RPM=1400$ to $RPM=2000$}
\label{fig:case1}

\end{figure}

\begin{figure}
    \centering
    \includegraphics[width=0.55\linewidth,height = 0.4\linewidth]{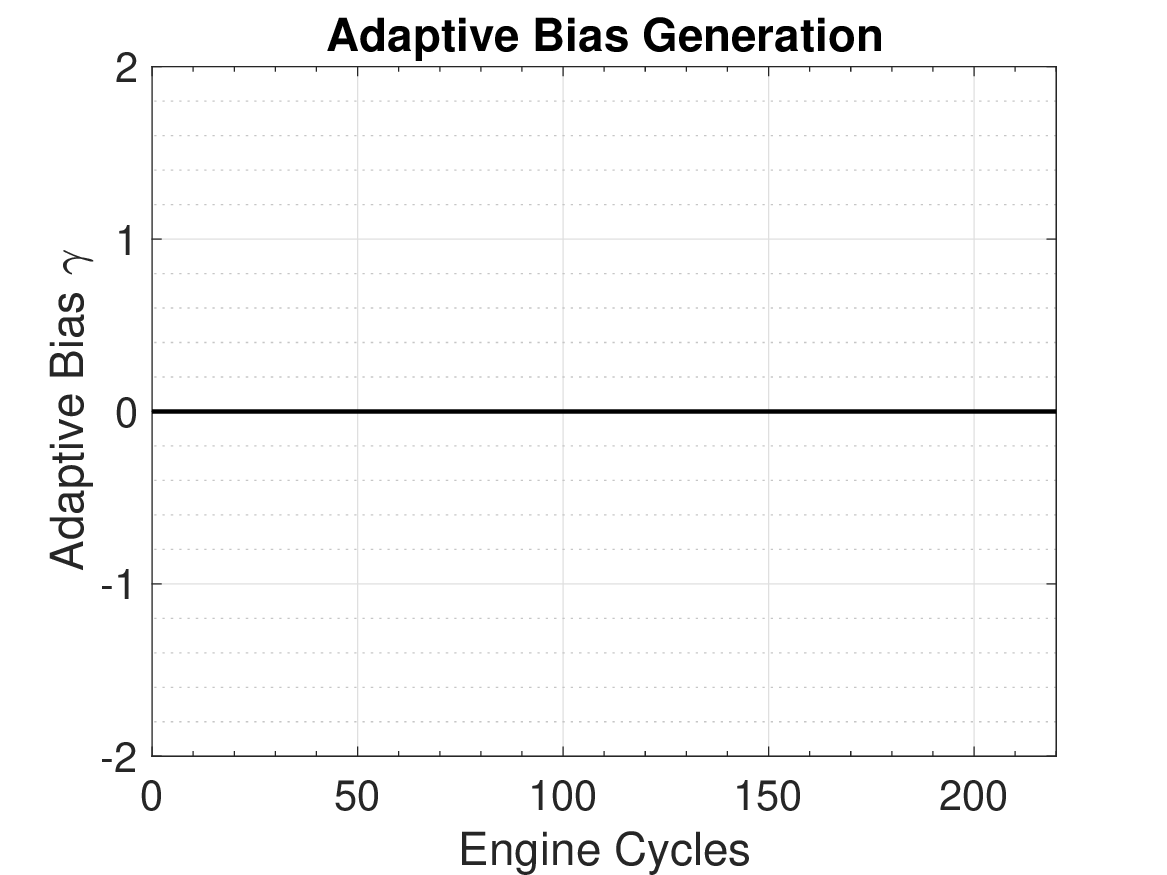}
    \caption{Adaptive Bias for Case 1}
    \label{fig:case1_bias}
\end{figure}

\begin{figure}[htp]
\centering

\begin{subfigure}{0.65\textwidth}
  \centering
  \includegraphics[width=0.85\linewidth,height = 0.55 \linewidth]{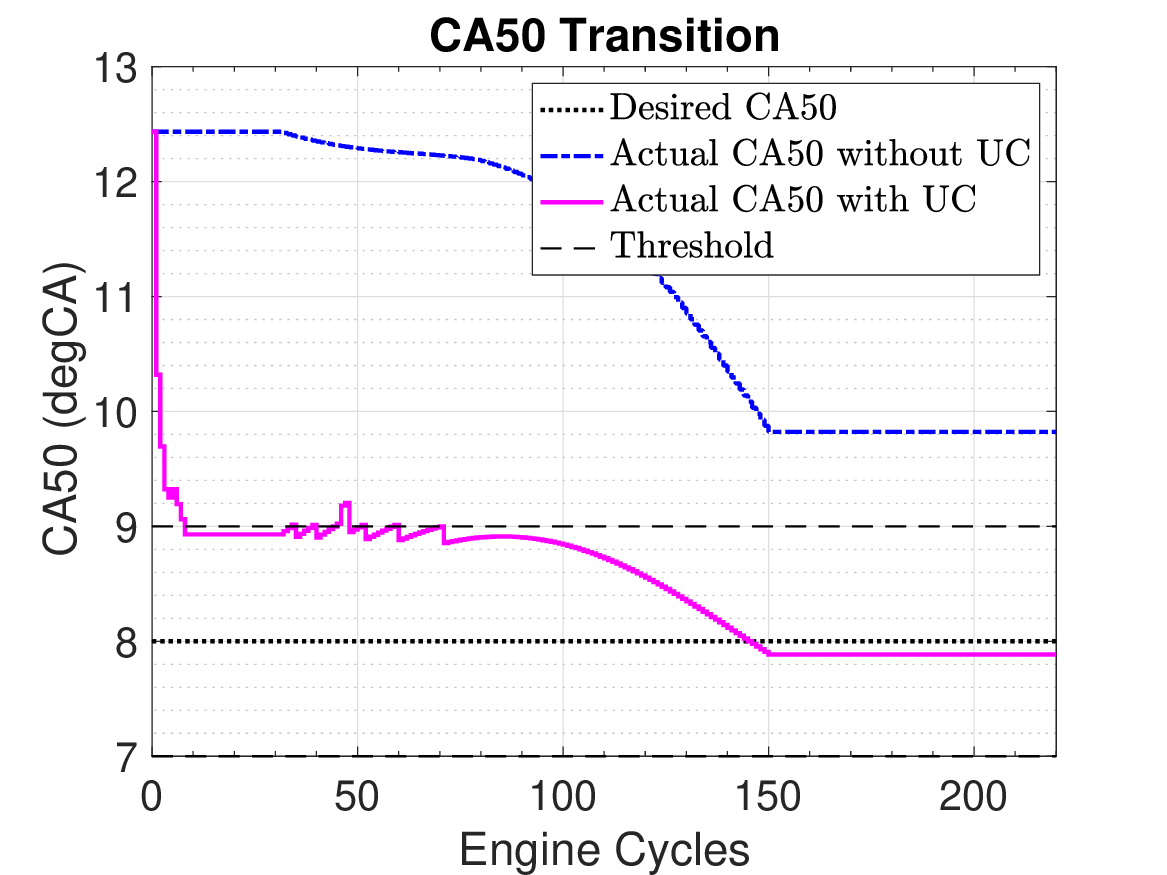}
\end{subfigure}


\begin{subfigure}{0.65\textwidth}
  \centering
  \includegraphics[width=0.85\linewidth,height = 0.55\linewidth]{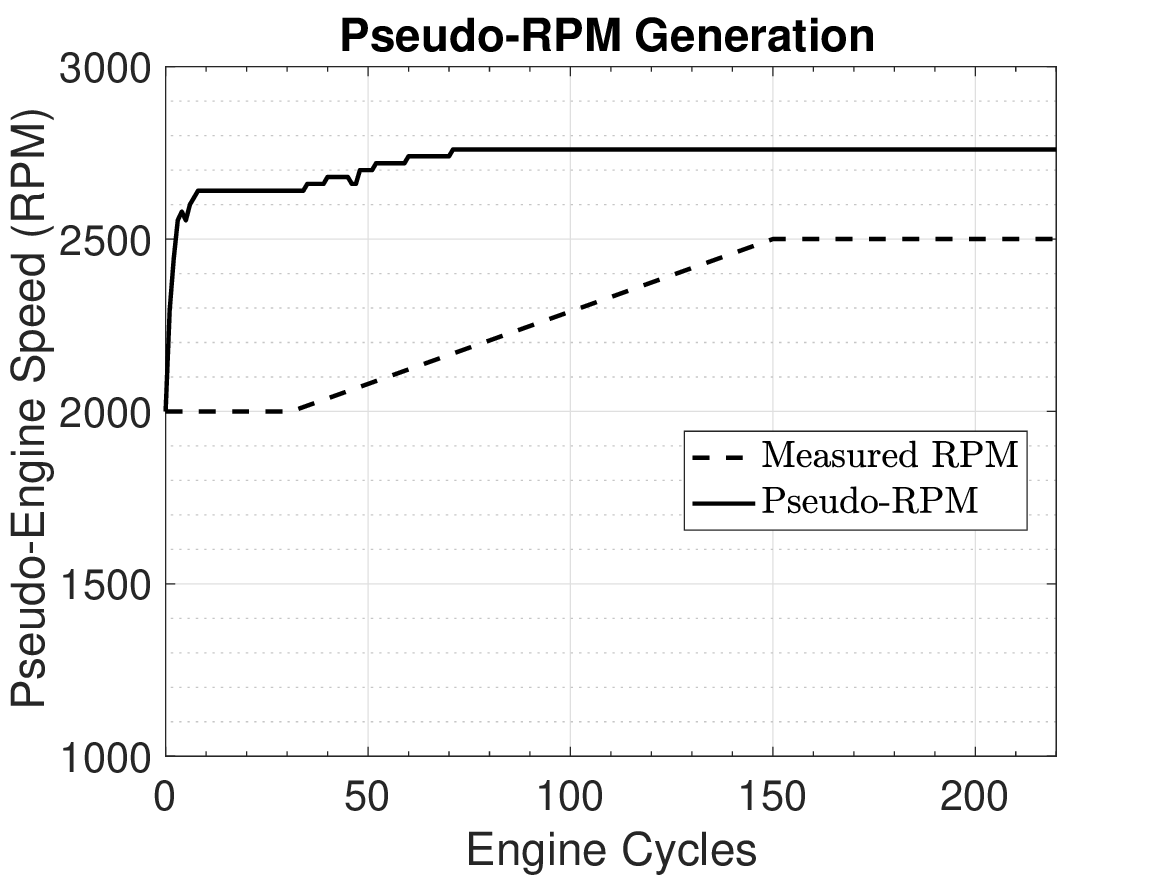}
\end{subfigure}


\begin{subfigure}{0.65\textwidth}
  \centering
  \includegraphics[width=0.85\linewidth,height = 0.55\linewidth]{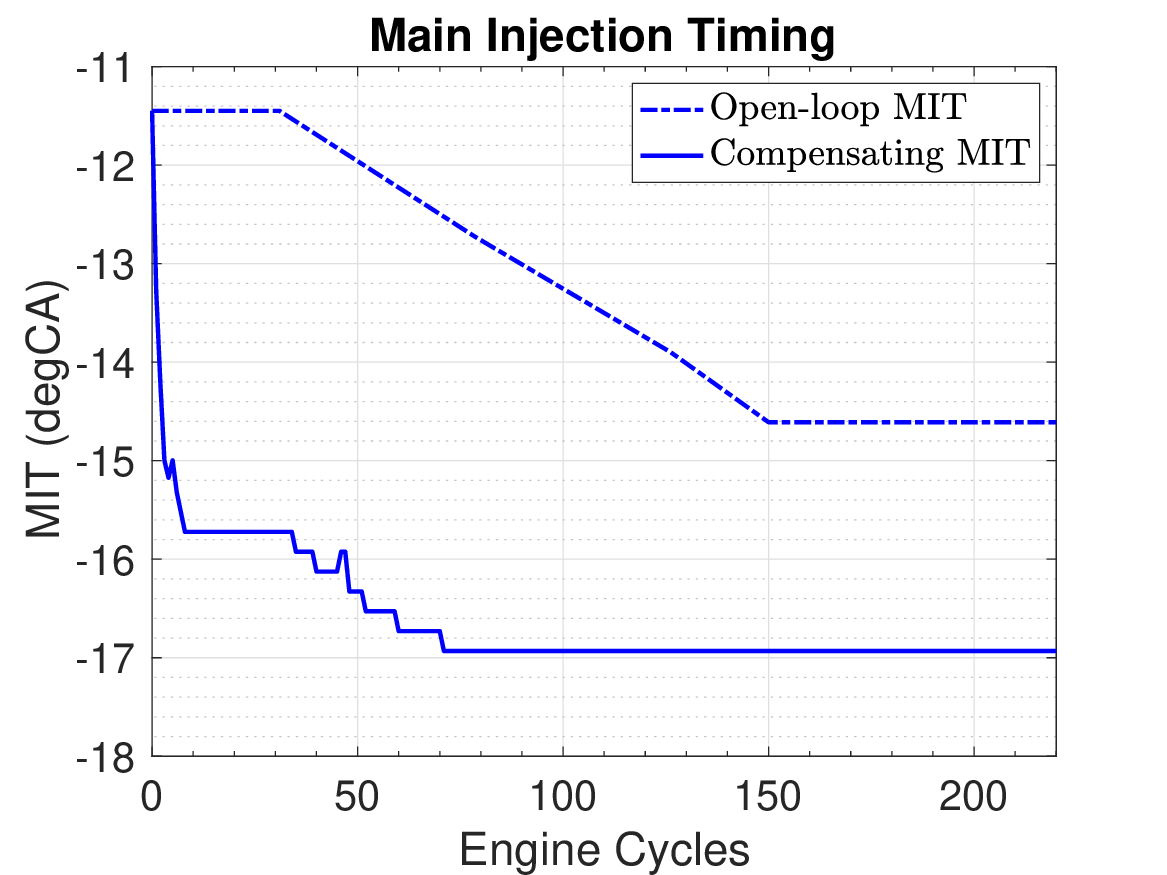}
\end{subfigure}


\begin{subfigure}{0.65\textwidth}
  \centering
  \includegraphics[width=0.85\linewidth,height = 0.55\linewidth]{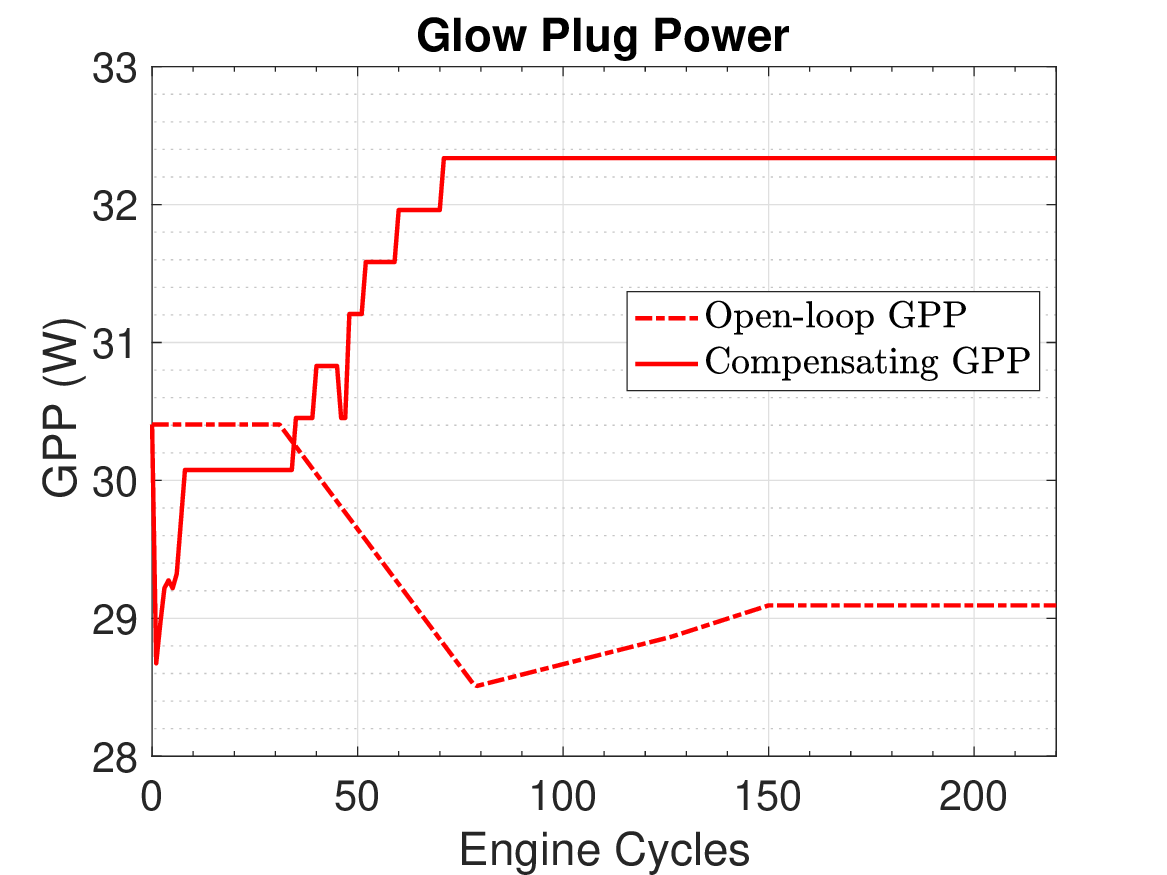}
\end{subfigure}

\caption{Simulation Results for Case 2 where CI engine is supplied with fuel $CN=36$ and engine speed changing from $RPM=2000$ to $RPM=2500$}
\label{fig:case2}

\end{figure}

\begin{figure}
    \centering
    \includegraphics[width=0.55\linewidth,height = 0.4\linewidth]{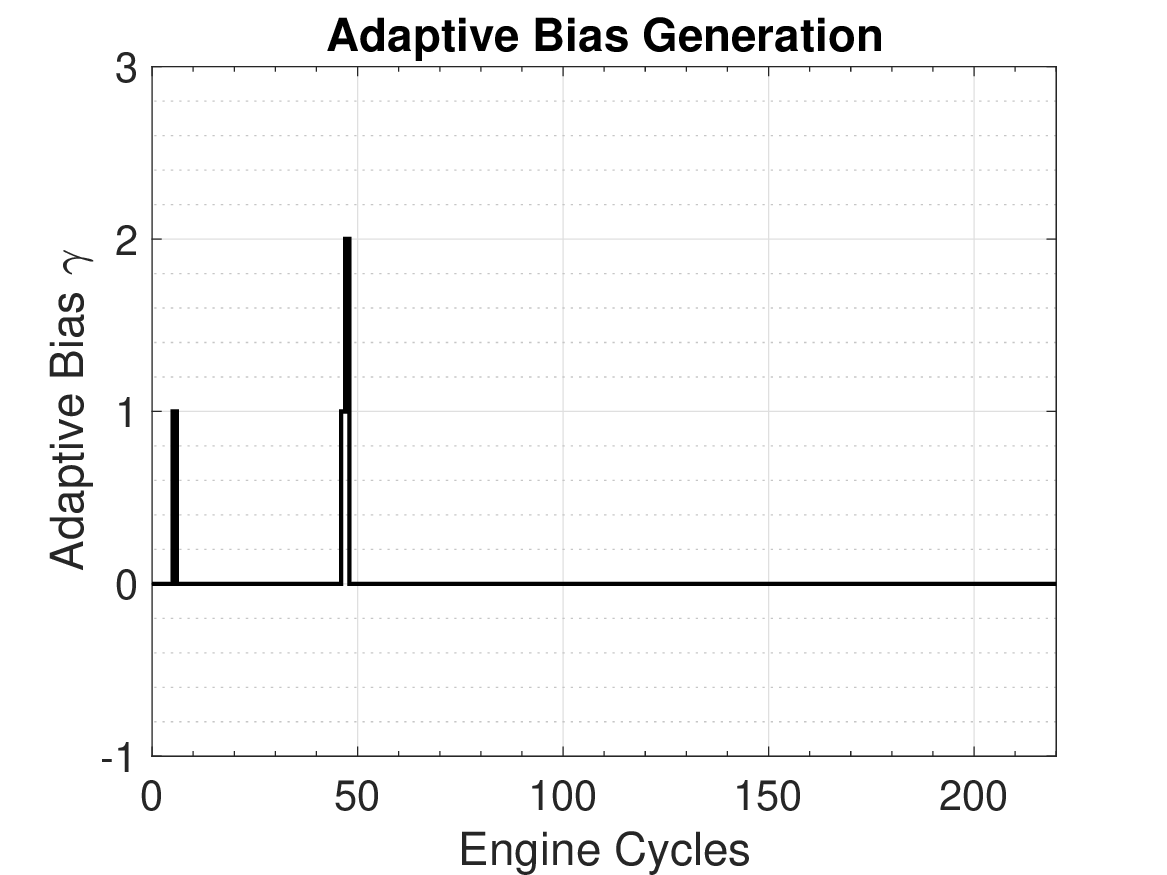}
    \caption{Adaptive Bias for Case 2}
    \label{fig:case2_bias}
\end{figure}

\subsection{Parameter Selection}

The convergence parameters are calibrated directly from the high-fidelity GP engine model by measuring how CA50 responds to a change in the pseudo-RPM coordinate. With cetane number and engine speed held fixed, the pseudo-RPM is perturbed by $\Delta Z_2$ over $\pm300$ RPM and the resulting CA50 change $\Delta y$ recorded which is the local sensitivity $|\Delta y|/|\Delta Z_2|$ (deg\,CA/RPM). This sweep is repeated across the operating envelope of the test cases ($CN\in\{36,45\}$, pseudo-RPM $2000$--$2800$ RPM, engine speed $1400$--$2500$ RPM which is 70 conditions in total), and each parameter is taken as the worst case over all of them.

The step radius $R=\delta_0=300$ RPM is the largest perturbation over which CA50 still moves toward the setpoint (the sensitivity keeps its sign), so a per-cycle correction up to $300$ RPM is always in the correct direction. The sensitivity stays between $\mu_{\min}=2.8\times10^{-3}$ and $\mu_{\max}=1.5\times10^{-2}$ deg\,CA/RPM. The strictly positive lower bound guarantees that every correction produces a measurable CA50 change (no dead zone), while the upper bound caps the step at $|e_\ell|/\mu_{\max}$ so a single correction cannot overshoot the setpoint. The applied step is therefore bounded as $Z_2^{*}=\min(\delta_0,\,|e_\ell|/\mu_{\max},\,R)$, and is floored at $\eta_{\min}=20$ RPM which is chosen below $\varepsilon_{\mathrm{thr}}/\mu_{\max}=66$ RPM so the floor itself cannot overshoot. Thus, one avoids stalling on negligibly small corrections. With a convergence tolerance $\varepsilon_{\mathrm{thr}}=1$ deg\,CA, these settings drive CA50 into the tolerance band in a finite number of cycles per Theorem~\ref{thm:convergence}.

The objective in this section is to validate the efficacy of the proposed uncertainty compensator in achieving desired engine performance i.e. CA50, under challenging real-time conditions where the engine is subjected simultaneously to fluctuations in fuel ignition quality (CN), variations in engine speed (RPM), and inherent model uncertainties. To perform the simulation tests, the desired CA50 is kept fixed at $8$ deg CA aTDC (after Top Dead Center). Two different cases are considered: (1) Case 1 where  engine speed changes from $RPM=1400$ to $RPM=2000$ with fuel $CN = 45$, and (2) Case 2 where  engine speed changes from $RPM=2000$ to $RPM=2500$ with fuel $CN = 36$. The threshold value is set to $1$ deg CA, in order to appropriately trigger the uncertainty compensator algorithm.

\subsection{Results and discussions}

The open-loop study in Figure \ref{fig:case1} and \ref{fig:case2} shows how the engine behaves under different fuel blends and engine speeds. It is clear that the open-loop actual $CA50$ (blue line) stays far from the desired value. This means the open-loop control is not able to achieve the desired combustion phasing. To understand why this happens, let us consider the sensitivity curves shown in Figure \ref{fig:sc} for $CN=36$. These curves illustrate the variation of $CA50$ with engine speed under fixed control inputs ($MIT = -7$, $GPP = 0$). The discrepancy between the real engine response and the GPR model prediction in this figure highlights the modeling error, which can be attributed to the limited training data available for the GPR model. As a result, the $CA50$ predicted by the model does not align well with the actual engine output across the full range of operating speeds. This mismatch causes tracking errors in the open-loop case. It is important to note that the level of model uncertainty depends on how far the engine is operating from the training points. For example for $CN=36$, when the engine operates at $RPM = 1400$, the model is relatively more accurate because this speed is closer to the training points. However, at $RPM = 2200$, which is farther from the training data, the model uncertainty increases, leading to larger errors.

Let us now examine how the uncertainty compensator handles model inaccuracies while running at varied operating conditions in both cases. Figure \ref{fig:case1} presents the system performance for Case 1, while Figure \ref{fig:case2} shows the results for Case 2. In both cases, the tracking error defined as the difference between the actual $CA50$ and the desired value, exceeds a pre-specified threshold. This condition activates the uncertainty compensator. Once triggered, the compensator begins searching along the model's sensitivity curve for given current $MIT$ and $GPP$. Its goal is to find a pseudo-speed value where the model output $CA50$ closely matches the actual engine $CA50$. After identifying this pseudo-speed, it is passed to the LUT to generate updated values of $MIT$ and $GPP$. These new control inputs are applied to the engine, leading to a reduction in tracking error. This correction process repeats across engine cycles and the tracking error continues to decrease until it falls below the threshold. For Case 1, it is observed that the error steadily decreases over successive engine cycles, indicating that the compensator effectively addresses the model uncertainties in this scenario. For Case 2, a similar process is followed initially. However, around certain engine cycles, the tracking error suddenly increases instead of decreasing, even though the compensator is still active. To handle this situation, the algorithm introduces an additional correction mechanism. Instead of trying to match the model output exactly to the actual $CA50$, the compensator now targets a slightly biased value, which is the actual $CA50$ plus a small, non-zero offset called the adaptive bias $\gamma$. This bias $\gamma$ allows the control input to be fine-tuned more effectively when the direct matching strategy fails. As a result, the error starts decreasing again and eventually converges to within the acceptable range of threshold. Figure \ref{fig:case2_bias} clearly show that $\gamma$ becomes non-zero during the adjustment phase, indicating that the compensator has switched strategies to maintain performance otherwise, it remains zero as in Figure \ref{fig:case1_bias}.

This study thus, demonstrates that the proposed data-driven control framework is robust to model uncertainties and capable of maintaining reliable system performance in real-time even under the effect of fluctuations in operating conditions of fuel ignition quality and engine speed. The uncertainty compensator plays a key role by dynamically adjusting the control inputs to mitigate the effect of such challenging scenarios. As observed in the results, the compensator generates modified control inputs (solid lines) that deviate from the baseline open-loop policy (dashed lines). These updated inputs drive the actual $CA50$ response (magenta line) to closely follow the desired $CA50$ trajectory.

\section{Conclusion} \label{sec:six}
This work presents a uncertainty compensation framework for achieving reliable combustion phasing in multi-fuel compression ignition engines where simultaneous fluctuations in fuel ignition quality, variations in engine speed and inherent model uncertainties adversely affect system performance. By introducing the concept of a pseudo-engine speed, the proposed method enables real-time adaptation of control inputs through uncertainty-aware feedforward compensation, synthesized via inversion of a Gaussian Process Regression-based surrogate model. The approach effectively steers combustion phasing toward desired CA50 value even under such challenging scenarios. Validation using data from high-fidelity engine simulations confirms the efficacy and generalizability of the framework. Further, this work includes the finite-time convergence through stability analysis of the proposed strategy. The results demonstrate that the integration of data-driven modeling with uncertainty compensation offers a direction for real-time combustion control in multi-fuel propulsion systems. Future work will aim to extend the proposed framework to efficiently incorporate a broader set of engine control parameters, thereby enhancing its scalability and applicability to more complex scenarios in higher dimensional setups.

\section*{CRediT Authorship Contribution Statement}

\textbf{Rajasree Sarkar:} Conceptualization, Data curation, Formal analysis, Investigation, Methodology, Software, Validation, Visualization, Writing – original draft, Writing – review and editing \\
\textbf{Arunava Banerjee:} Data curation, Formal analysis, Methodology, Software, Visualization, Writing – original draft, Writing – review and editing \\
\textbf{Sathya Aswath Govind Raju:} Formal analysis, Software, Validation, Visualization, Writing – review and editing \\
\textbf{Ishan Berk Altiner:} Formal analysis, Software, Validation, Visualization, Writing – review and editing\\
\textbf{Zongxuan Sun:} Conceptualization, Funding acquisition, Investigation, Methodology, Project administration, Resources, Supervision, Writing – review and editing \\
\textbf{Kenneth Kim:} Project administration, Resources, Supervision, Writing – review and editing \\
\textbf{Chol-Bum Mike Keown:} Project administration, Resources, Supervision, Writing – review and editing 

\section*{Declaration of competing interest}
The authors declare that they have no known competing financial interests or personal relationships that could have appeared to influence the work reported in this paper.

\section*{Acknowledgements}
Research was sponsored by the DEVCOM Army Research Laboratory and was accomplished under Cooperative Agreement Number W911NF-20-2-0161. The views and conclusions contained in this document are those of the authors and should not be interpreted as representing the official policies, either expressed or implied, of the DEVCOM Army Research Laboratory of the U.S. Government. The U.S. Government is authorized to reproduce and distribute reprints for Government. The authors would like to thank the team at Engine Research center of University of Wisconsin Madison for providing us with the experimental data for this work. The authors would also like to mention that this work used computational resources at the Minnesota Supercomputing Institute (MSI), University of Minnesota. The authors would also like to thank the team at Computational Reactive Flow \& Energy Lab (CRFEL) of University of Minnesota–Twin Cities, Minneapolis for providing us with the CFD simulated data for this work.

\section*{Declaration of conflicting interests}
The authors declared no potential conflicts of interest with respect to the research, authorship, and/or publication of this article.

\section*{Funding}
Research was sponsored by the DEVCOM Army Research Laboratory and was accomplished under Cooperative Agreement Number W911NF-20-2-0161. The views and conclusions contained in this document are those of the authors and should not be interpreted as representing the official policies, either expressed or implied, of the DEVCOM Army Research Laboratory of the U.S. Government. The U.S. Government is authorized to reproduce and distribute reprints for Government.

\bibliographystyle{elsarticle-num}
\bibliography{Ref_BIO}
\appendix
\section{} \label{appendix}

\subsection{Notation and Setup.} 
Consider Algorithm \ref{alg:corrected} operating at a fixed real operating point $Z_0=(Z_{10},Z_{20})$. Let $y_{d}\in\mathbb{R}$ denote the desired CA50 setpoint. Define $e(k)=y(k)-y_{d}$ as the CA50 error at cycle~$k$, and let $e_{\ell}=y^{prev}_{\ell}-y_{d}$ denote the error at the $\ell$-th accepted anchor. Let $V_{\ell}=1/2~e_{\ell}^2$. For a given anchor $(y^{prev},Z^{prev})$, let $\Delta Z_2(\gamma)=Z^{new}(\gamma)-Z^{prev}$ denote the effective pseudo-RPM perturbation produced by the optimization step at bias value~$\gamma$.
Observe that the quantities $T_1=y^{cur}-y^{prev}$ and $T_2=y^{cur}+y^{prev}-2y_{d}$, therefore, 
\begin{align}\label{observation}
    T_1*T_2 &= \bigl[y^{cur}-y_{d}\bigr]^2-\bigl[y^{prev}-y_{d}\bigr]^2 \nonumber \\
    &=\;(e^{cur})^2-(e^{prev})^2
\end{align}
Thus, $T_1*T_2<0$ if and only if $|y^{cur}-y_{d}|<|y^{prev}-y_{d}|$.

\subsection{Assumptions.}\label{appendixB}

\noindent\textit{(B1) Sector bound on the plant--LUT composition:} There exist constants $0<\mu_{\min}\le\mu_{\max}$ such that, for every anchor $Z^{prev}$ visited by the algorithm and every $\Delta Z_2$ in a bounded convex neighborhood~$\mathcal{N}$ of the origin,
\begin{align}
    \mu_{\min}\,|\Delta Z_2|\;&\le\;|f(Z_{10},Z^{prev}+\Delta Z_2)-f(Z_{10},Z^{prev})|\;\; \nonumber \\
    &\le \mu_{\max}\,|\Delta Z_2|.
\end{align}

\noindent\textit{(B2) Local descent direction from optimization:} For any anchor with $e^{prev}\neq 0$ and for all $\gamma$ with $0<|\Delta Z_2(\gamma)|\le\delta_0$, the perturbation $\Delta Z_2(\gamma)$ produced by the optimization step is such that $\text{sign}\bigl(f(Z_{10},Z^{prev}+\Delta Z_2(\gamma))-f(Z_{10},Z^{prev})\bigr)=-\text{sign}(e^{prev})$. That is, the plant CA50 moves toward~$y_{d}$ under the model-based correction.

\medskip
\noindent\textit{(B3) Regularity of the identification mapping:} The mapping $\gamma\mapsto \Delta Z_2(\gamma)$ is continuous on $[-\gamma_{\max},\gamma_{\max}]$ and locally Lipschitz with constant $L_d>0$ in a neighborhood of any $\gamma$ at which $\Delta Z_2(\gamma)$ lies in the descent region.

\medskip
\noindent\textit{(B4) Effective step-size control:} For any anchor with $|e^{prev}|\ge\varepsilon_{\mathrm{thr}}$, let $Z_2^*=\min(\delta_0,\,|e_{\ell}|/\mu_{\max},\,R)$. The sequence of effective steps $\{\Delta Z_2(\gamma_n)\}$ generated by the inner loop satisfies $0<|\Delta Z_2(\gamma_n)|\le Z_2^*$ for every inner-loop cycle and lies along the descent direction of~B2, so a qualifying step is reached in finitely many cycles.

\medskip
\noindent\textit{(B5) Quantitative descent at acceptance:} There exists a non-decreasing function $\eta:(0,\infty)\to(0,\infty)$ with $\eta(s)>0$ for all $s>0$, such that whenever the inner loop terminates at anchor~$\ell$, the accepted perturbation satisfies $|\Delta Z_2^{(\ell)}|\ge\eta(|e_{\ell}|)$. In Algorithm~\ref{alg:corrected} this holds with the constant choice $\eta(\cdot) \equiv \eta_{min}$ enforced by the Step-4 minimum-step floor.

\subsection{Proof to Theorem \ref{thm:convergence}}\label{thm_proof}
(i) Fix anchor~$\ell$ with $|e_{\ell}|>0$. Set $Z_2^*=\min(\delta_0,|e_{\ell}|/\mu_{\max})$. For any perturbation $\Delta Z_2$ along the descent direction with $|\Delta Z_2|\in(0,Z_2^*]$, Assumption~B1 gives $|\Delta y|\in[\mu_{\min}|\Delta Z_2|,\mu_{\max}|\Delta Z_2|]$ with $\mu_{\max}|\Delta Z_2|\le|e_{\ell}|$, and Assumption~B2 gives $\text{sign}(\Delta y)=-\text{sign}(e_{\ell})$. Therefore $|e_{new}|=|e_{\ell}|-|\Delta y|<|e_{\ell}|$, whence $T_1 * T_2=e_{new}^2-e_{\ell}^2<0$ by \eqref{observation}. By Assumption~B4, the inner-loop sequence of effective steps visits $(0,Z_2^*]$ in finite iterations. At the first such visit with the step in the descent direction (guaranteed by~B2), $T_1 * T_2<0$ is triggered, and the inner loop terminates. 
\medskip
(ii) By Part~(i), the anchor is updated with $T_1 * T_2<0$, so \eqref{observation} gives $e_{\ell+1}^2<e_{\ell}^2$, i.e., $V_{\ell+1}<V_{\ell}$.
\medskip
(iii) \textit{Finite-time practical convergence.} While the algorithm operates, the loop guard enforces $|e_{\ell}|\ge\varepsilon_{\mathrm{thr}}$ for $\ell=0,1,\dots,\ell^{*}-1$. At each such anchor, the no-overshoot bound $|\Delta y_{\ell}|\le\mu_{\max}|\Delta Z_2^{(\ell)}|\le\mu_{\max}Z_2^*\le|e_{\ell}|$ gives the exact Lyapunov decrease
\begin{align}
V_{\ell}-V_{\ell+1}=\tfrac{1}{2}\,|\Delta y_{\ell}|\bigl(2|e_{\ell}|-|\Delta y_{\ell}|\bigr).
\end{align}
By B5 and B1, $|\Delta y_{\ell}|\ge\mu_{\min}\eta(|e_{\ell}|)\ge\mu_{\min}\eta(\varepsilon_{\mathrm{thr}})$, and since $\mu_{\min}\eta(\varepsilon_{\mathrm{thr}})\le|e_{\ell}|$ with $|e_{\ell}|\ge\varepsilon_{\mathrm{thr}}$ we have $2|e_{\ell}|-|\Delta y_{\ell}|\ge |e_{\ell}| \ge \varepsilon_{thr}$. Hence
\begin{align*}
V_{\ell}-V_{\ell+1} \geq & \tfrac{1}{2}\,\mu_{\min}\eta(\varepsilon_{\mathrm{thr}})\bigl(2\varepsilon_{\mathrm{thr}}-\mu_{\min}\eta(\varepsilon_{\mathrm{thr}})\bigr)\\=: & ~\delta(\varepsilon_{\mathrm{thr}})>0,
\end{align*}
a constant independent of~$\ell$. Summing, $V_0-V_{\ell^{*}}\ge\ell^{*}\,\delta(\varepsilon_{\mathrm{thr}})$, and since $V_{\ell^{*}}\ge0$ we obtain $\ell^{*}\le\lfloor V_0/\delta(\varepsilon_{\mathrm{thr}})\rfloor+1<\infty$, with $|e_{\ell^{*}}|<\varepsilon_{\mathrm{thr}}$ at termination. \qed

\begin{remark}[Physical interpretation of Assumptions B1--B5]
Assumption~B1 reflects the standard physical property that CA50 varies continuously and with bounded nonzero sensitivity to changes in the combustion-phasing actuator (here encoded through the LUT evaluated at pseudo-RPM coordinates). This is consistent with sensitivity studies for CA50 in engines and holds in well-calibrated operating regions away from misfire or knock boundaries. Assumption~B2 requires that the GP model's qualitative directional information (which way to adjust pseudo-RPM to move CA50 toward~$y_{d}$) is correct at the plant. This is a minimal model fidelity condition that does not require accurate quantitative prediction. Assumption~B3 is a regularity condition on the optimization problem that follows from smoothness of the GP model and positivity of the regularization parameter~$\beta$. It is a technical condition satisfied by construction in any implementation using a smooth GP model with $\beta>0$. Assumption~B4 requires that every pseudo-RPM correction applied to the plant is a conservative (small) step lying in the favorable region. The perturbation safeguard of Algorithm~\ref{alg:corrected} (Step~4) enforces $|\Delta Z_2|\le Z_2^*$ at every cycle by construction via $\gamma$-bisection toward zero, with a clamping fallback when necessary. As a consequence, B4 holds. Combined with the bounded range $\gamma\in[-\gamma_{\max},\gamma_{\max}]$ and the regularization in the optimization problem, this guarantees that each accepted correction stays within the sector-bounded region. Assumption~B5 holds by construction of the minimum-step floor in Algorithm~\ref{alg:corrected} where Step~4, enforces $\left|\Delta Z_2^{(\ell)}\right| \geq \eta_{\min} > 0$ at every applied update. This condition serves as the lower-bound counterpart of the maximum-step safeguard and precludes vanishingly small corrections that could otherwise stall finite-time convergence.
\end{remark}
\end{document}